\documentclass[11pt]{article}
\usepackage{amsfonts}
\usepackage{amsfonts}
\usepackage{amsfonts}
\usepackage{amsmath}
\usepackage{mathrsfs}
\usepackage{amsmath}
\usepackage{mathrsfs}
\usepackage{mathrsfs}
\usepackage{amsmath}
\usepackage{fancyhdr}
\usepackage{mathrsfs,epsfig,morefloats}
\usepackage{amsfonts}
\usepackage{natbib}
\usepackage{amsmath}
\usepackage{amsthm}
\usepackage{amssymb}
\usepackage{color}
\usepackage{lscape}
\usepackage{rotating}
\usepackage{caption2}
\usepackage{subfigure}
\usepackage{graphicx}

\graphicspath{ {C:/Users/jean/Desktop/NUS/} }

\usepackage{amssymb}

\newcommand{\Date}[1]{\def\@Date{#1}}
\def\today{\number\day~\ifcase\month\or
 January\or February\or March\or April\or May\or June\or
 July\or August\or September\or October\or November\or December\fi~\number\year}

\def\be{\begin{equation}}
\def\ee{\end{equation}}
\def\bea{\begin{eqnarray}}
\def\eea{\end{eqnarray}}
\def\bd{\begin{displaymath}}
\def\ed{\end{displaymath}}
\def\bda{\begin{eqnarray*}}
\def\eda{\end{eqnarray*}}
\def\bsm{\begin{small}}
\def\esm{\end{small}}

\def\t0{\theta_0}

\def\ha1{\hat \beta_1}

\def\bnt{\begin{enumerate}}
\def\ent{\end{enumerate}}
\def\T{{ \mathrm{\scriptscriptstyle T} }}

\def\AS{A\"{\i}t-Sahalia}

\def\bsc{\begin{scriptsize}}
\def\esc{\end{scriptsize}}

\newtheorem{theorem}{Theorem}

\newtheorem{lemma}{Lemma}

\theoremstyle{definition}

\newcommand\independent{\protect\mathpalette{\protect\independenT}{\perp}}
\def\independenT#1#2{\mathrel{\rlap{$#1#2$}\mkern2mu{#1#2}}}

\makeatletter
\newcommand{\figcaption}{\def\@captype{figure}\caption}
\newcommand{\tabcaption}{\def\@captype{table}\caption}
\makeatother

\newcommand{\cov}{{\rm Cov}}

\newcommand{\etal}{\mbox{\sl et al.\;}}

\newcommand{\eg}{\mbox{\sl e.g.\;}}

\newcommand{\bB}{{\mathbf B}}

\newcommand{\bI}{{\mathbf I}}

\newcommand{\bX}{{\mathbf X}}
\newcommand{\bY}{{\mathbf Y}}
\newcommand{\bZ}{{\mathbf Z}}

\newcommand{\bu}{{\mathbf u}}

\newcommand{\bw}{{\mathbf w}}

\newcommand{\bz}{{\mathbf z}}

\newcommand{\bbeta}  {\boldsymbol{\beta}}

\newcommand{\blambda}{\boldsymbol{\lambda}}

\newcommand{\bepsilonbb}{\boldsymbol{\epsilon}}
\newcommand{\bOmega}{\boldsymbol{\Omega}}

\newcommand{\bSigma}{\boldsymbol{\Sigma}}
\newcommand{\bDelta}{\boldsymbol{\Delta}}
\newcommand{\bgamma}{\boldsymbol{\gamma}}

\newcommand{\bPsi} {\boldsymbol{\Psi}}

\newcommand{\btheta} {\boldsymbol{\theta}}

\newcommand{\bmu} {\boldsymbol{\mu}}

\newcommand{\bGamma} {\boldsymbol{\Gamma}}

\newcommand{\E}{\rm E}
\newcommand{\V}{\rm Var}

\def\JRSSB{{\sl Journal of the Royal Statistical Society}, {\bf B}}

\def\JASA{{\sl Journal of the American Statistical Association}}

\def\AS{{\sl The Annals of Statistics}}

\renewcommand\arraystretch{1.1}
\usepackage[pdftex]{hyperref}

\begin{document}

\title{ \bf Sufficient Dimension Reduction for Classification}

%\author{xx\\xxxxx \and xx\\xxxxx }
\author{Xin Chen\footnote{Department of Statistics and Applied Probability, National University of Singapore, Singapore {\rm 117546}, Singapore; stacx@nus.edu.sg}, Jingjing Wu\footnote{Department of Mathematics and Statistics, University of Calgary, Calgary, Alberta, Canada; jinwu@ucalgary.ca}, Zhigang Yao\footnote{Department of Statistics and Applied Probability, National University of Singapore, Singapore {\rm 117546}, Singapore; zhigang.yao@nus.edu.sg}, Jia Zhang\footnote{Corresponding author. Department of Statistics, Southwestern University of Finance and Economics, Chengdu {\rm 611130}, China; jeanzhang9@2015.swufe.edu.cn}}
%\date{\large The University of Melbourne, Temple University, and University of Rochester}

%\date{}
\maketitle

\begin{abstract}
We propose a new sufficient dimension reduction approach designed deliberately for high-dimensional classification. This novel method is named maximal mean variance (MMV), inspired by the mean variance index first proposed by \cite{Cui2015}, which measures the dependence between a categorical random variable with multiple classes and a continuous random variable. Our method requires reasonably mild restrictions on the predicting variables and keeps the model-free advantage without the need to estimate the link function. The consistency of the MMV estimator is established under regularity conditions for both fixed and diverging dimension ($p$) cases and the number of the response classes can also be allowed to diverge with the sample size $n$. We also construct the asymptotic normality for the estimator when the dimension of the predicting vector is fixed. Furthermore, %although without any definite theoretical proof,
our method works pretty well when $n < p$. The surprising classification efficiency gain of the proposed method is demonstrated by simulation studies and real data analysis.

\end{abstract}

\begin{quote}
\noindent
{\sl Keywords}: Classification; Mean variance index; Sufficient dimension reduction; Consistency. \end{quote}

\begin{quote}
\noindent
{\sl MSC2010 subject classifications}: Primary 62H30; secondary 62G05
\end{quote}

\thispagestyle{empty}
\pagenumbering{gobble}

\newpage
\pagenumbering{arabic}

\setcounter{page}{1}

\section{Introduction}
\label{s1}

Sufficient dimension reduction fits into what is currently quite a hot area in research of high dimensional data. Large quantities of related articles and studies have appeared in recent decades. However, most of the literature focuses on the regression problem where the responsible variable $Y$ is a continuous scalar, while little is designed specially for the problem of classification with a categorical response. The slice-based methods, including but not limited to the seminal sliced inverse regression (SIR; \cite{Li1991}), sliced average variance estimation (SAVE; \cite{CookWeisberg1991}), directional regression (DR; \cite{LiWang2007}) and sliced regression (SR; \cite{WangXia2008}), can be naturally applied to the classification problem with the slices determined directly by the classes of the response. It seems to work nicely but the number of the slices is restricted by the number of the classes which can be problematic when there are only a few categories. More specifically, when the response is a binary variable, the number of the slices is imposed as $2$ and the number of effective dimension-reduction directions is correspondingly forced to be $1$, which would directly reduce the accuracy of classification. Moreover, almost all of the above methods require the linearity condition or constant covariance condition, or both, which are difficult to verify in practice, and the results may be misleading if these conditions are violated. Other popular sufficient dimension reduction methods, like minimum average variance estimation (MAVE; \cite{Xia2002}), inverse regression (IR; \cite{CookNi2005}) and distance covariance based sufficient dimension reduction (DC; \cite{ShengYin2013}, \cite{ShengYin2016}) either require the responsible variable to be continuous or treat the response as a numeric variable. Therefore, from the perspective of sufficient dimension reduction, it's necessary to do some work deliberately into the case of the categorical responsible variable or in other words, the classification.

More importantly, from the perspective of classification itself, dimension reduction is of importance for constantly emerging high-dimensional classification problems considering the ``curse of dimensionality'' which appears in most classification approaches, such as the frequently used linear and quadratic classifiers, support vectors machines, k-nearest neighbours, decision trees, neural networks and new methods like distance weighted discrimination (\cite{Marron2007}, \cite{Marron2015} and \cite{WangZou2018}) and so on. The usual practice is to conduct projection or variable selection to reduce dimensionality as a first step. The projection methods have been widely applied to classification for the gene expression data. Related research includes but is not limited to principle component analysis in \cite{Zou2006}, \cite{Bair2006} and \cite{Shao2014}; sliced inverse regression in \cite{BuraPfeiffer2003} and \cite{Antoniadis2003}; and partial least squares in \cite{Boulesteix2004}. Of note is the fact that sufficient dimension reduction is indeed a kind of projection methods where all the information related to classification is preserved. Variable selection is another line of dimension reduction approach. It can deal with classification problems with extremely large dimensionality. See, for example, the nearest shrunken centroids method in \cite{Tibshirani2002}, the features annealed independence rules in \cite{FanFan2008} and the latest research, including \cite{AndrewsMcNicholas2014}, \cite{Stefanski2014}, and \cite{PartoviDavison2015}, among others. Furthermore, \cite{Cui2015} used a novel mean variance (MV) index to implement model-free feature screening for high dimensional discriminant analysis, which is related closely to the proposed method.

In addition to the classification methods aforementioned, there is also a great deal of other research work on classification. To name but a few, \cite{DettlingBuhlmann2003} and \cite{BuhlmannYu2003} studied boosting under logit loss, and $L_2$ loss in the high-dimensional setting, respectively; \cite{GreenshteinRitov2004} introduced the concept of persistence, which is weaker than consistency and pays more attention to misclassification error; \cite{DonohoJin2008} employed higher criticism thresholding for feature screening when the useful features are both rare and weak; and following their work, \cite{Fan2013} proposed a two-stage classification procedure based on innovated thresholding and high criticism thresholding in the sparse Gaussian graphical model.

In this article, we propose a novel sufficient dimension reduction approach -- maximum mean variance(MMV), designed deliberately for a high-dimensional classification problem based on the mean variance index first proposed by \cite{Cui2015}. This method is not slice-based, thus it circumvents the restriction on slice number. Moreover, the approach does not require the linearity condition or constant variance condition, nor does it require any special distributions on the predicting vector $\bX$, $\bX|Y$ or $Y|\bX$, which is essential in the methods of \cite{ZhuZeng2006}, \cite{CookForzani2009}, \cite{CookLi2009}, \cite{BuraForzani2015}, \cite{Bura2016} and \cite{Zhang2018}. In addition, our method keeps the model-free advantage without the need to estimate the link function. These advantages broaden the scope of the application of our method. The consistency of the MMV estimator is established for both fixed and diverging dimension cases and asymptotic normality is constructed for the case of fixed dimension. The relationship between MMV and classification is more than the usual stepwise heuristics of dimension reduction first and classification next, which is elaborated upon by taking the examples of linear discriminant analysis (LDA) and index models. Moreover, the asymptotic theory of MMV estimator is quite challenging to set up, because the empirical MV index includes the kernel estimation of conditional and unconditional distribution functions, thus it can not be directly expressed by the sum of independent and identically distributed random variables.

The rest of the paper is organised as follows. In Section \ref{s2}, we propose a new sufficient dimension reduction approach (MMV) to high-dimensional classification. In Section \ref{s4}, we elaborate on the delicate relationship between the MMV method and classification. Consistency and the asymptotic normality of the MMV estimator are studied in Section \ref{s3}. Several simulation studies together with numerical comparisons and a real data example are conducted to illustrate the efficiency and priority of the proposed method in Section \ref{s5}. Section \ref{s6} concludes the article and the technical proofs are deferred to the Appendix.

\setcounter{equation}{0}
\section{Method}
\label{s2}

\subsection{Mean variance index}

Let $Y$ be a categorical response with $R$ classes $\{y_1, y_2, \dots, y_R\}$ and $Z$ be a continuous covariate. The Mean Variance index (\cite{Cui2015}) is defined by:
\be \label{111}
MV(Z|Y)={\E}_{Z}[{\V}_{Y}(F(Z|Y))]=\sum_{r=1}^{R} p_r \int [F_r(z)-F(z)]^2dF(z)
\ee
where $F(z|Y)= \mathbb{P}(Z \le z|Y)$, $F(z)=\mathbb{P}(Z \le z)$, $F_{r}(z)=\mathbb{P} (Z \le z|Y=y_r)$ and $p_r=\mathbb{P}(Y=y_r)$ for $r=1, \dots, R$. It has been verified that $MV(Z|Y)=0$ if and only if $Y$ and $Z$ are independent. It is worth noting that the MV index characterizes both linear and nonlinear correlations between categorical variable $Y$ and continuous variable $Z$.

Let $\{(Y_i, Z_i): 1 \le i \le n\}$ be an i.i.d random sample of size $n$. Let $\hat{F}(Z)$ and $\hat{F}_r(Z)$ be some sample estimations of $F(Z)$ and $F_r(Z)$. The MV index can be estimated by:
\be \label{mv}
\widehat{MV}(Z|Y):=MV_n(Z|Y)=\frac{1}{n} \sum_{r=1}^{R} \sum_{i=1}^{n} \hat{p}_{r}[\hat{F}(Z_i)-\hat{F}_r(Z_i)]^2
\ee
where $\hat{p}_{r}=1/n \sum_{i=1}^{n} \mathrm{I}\{Y_i=y_r\}$ with $\mathrm{I}(\cdot)$ being the indicator function. \cite{Cui2015} used the empirical distributions of $Z$ and $Z|Y$ as their sample estimators in a screening procedure.

\subsection{Maximum Mean Variance method for sufficient dimension reduction}
%Maximum Mean Variance method (MMV).
Classification is a crucial statistical problem which has attracted interest for decades. Let $Y$ be the categorial response defined above and $\bX\in \mathbb{R}^p$ be continuous predictors. %In the high dimensional data, statisticians often use various screening methods in the hope of reducing the number of predictors to be smaller than the sample size. Unfortunately, in the reality, it may not always be so.
In this paper, %we propose an efficient and robust method in discriminant analysis which can be implemented in high dimension cases without applying screening procedure.
we propose a novel sufficient reduction approach designed deliberately for high-dimensional classification problems based on the mean variance index. The idea is to find a few linear combinations (or indexes) of original predictors that %are most important
contribute to classification without a loss of information, so that these derived indexes can then be utilized for classification. This is achieved through a stepwise maximization procedure of the MV index. Recall that $MV(Z|Y)=0$ if and only if $Z$ and $Y$ are statistically independent. Thanks to this property, the MV index is used for marginal feature screening in discriminant analysis (\cite{Cui2015}). Our novel idea is to abandon this and, on the contrary, we seek a $\bbeta \in \mathbb{R}^p$ such that $MV(\bbeta^\T \bX|Y)$ achieves its maximum under some constraints. This is why we named this method Maximum Mean Variance. Hereafter, we refer to the MMV together with the following classification approach as "MMV+.". A sequential algorithm is elaborated as follows: we find the first linear combination of the predictors from
\begin{equation*}
\bbeta_{01}=\arg \max_{\bbeta_1} MV(\bbeta_1^\T\bX|Y)
\end{equation*}
subject to $\bbeta_1^\T \bbeta_1 = 1$. Then the $k$th linear combination can be calculated from

\begin{equation}\label{BetaPop}
\bbeta_{0k}=\arg \max_{\bbeta_k} MV(\bbeta_k^\T\bX|Y)
\end{equation}
subject to $\bbeta_k^\T \bbeta_k = 1$ and $[\bbeta_{01},\cdots,\bbeta_{0(k-1)}]^\T \bbeta_k = \mathbf{0}$. %(or $[\hat{\beta}_1,\cdots,\hat{\beta}_{k-1}]^T \hat{\bSigma}\beta_k = 0$).
We continue this process till the MV index reaches $0$. This procedure is indeed conducting sufficient dimension reduction, which can be seen clearly from the following theorem.%it reaches some stopping rule, for example, smallest ten fold prediction error using $\hat{\beta}_1X, \cdots,\hat{\beta}_dX$ in LDA. We name our method as MMV (maximum mean variance).

\begin{theorem}\label{t1}
Suppose there exists a positive integer $d<p$ such that $MV(\bbeta_{01}^{\T}\bX|Y)\ge MV(\bbeta_{02}^{\T}\bX|Y) \ge \dots \ge MV(\bbeta_{0d}^{\T}\bX|Y)>0=MV(\bbeta_{0(d+1)}^{\T}\bX|Y)=\dots=MV(\bbeta_{0p}^{\T}\bX|Y)$ where $\bbeta_{0i}^{\T}\bbeta_{0i}=1$ and $\bbeta_{0i}^{\T}\bbeta_{0j}=0$ for $i,j=1,\dots,p$ and $i \ne j$. Then
\[
S_{Y|\bX}\subseteq {\rm span}\{\bbeta_{01},\dots,\bbeta_{0d}\},
\]
where $S_{Y|\bX}$ denotes the central dimension reduction subspace (\cite{Cook1994}, \cite{Cook1996}) and for any $k<d$,
\[
{\rm span}\{\bbeta_{01},\dots,\bbeta_{0k}\}\nsupseteq S_{Y|\bX}.
\]
\end{theorem}

Note that the subscript $0$ in $\bbeta_{0i}$'s in Theorem \ref{t1} is used to indicate that these variables are in the population level. The existence of $d<p$ is validated in the classical linear discriminant analysis (LDA) and in the index model setting. See more details in Section \ref{s4}.

In practice, we don't know the population MV index for any given $\bbeta$ and have to use $\widehat{MV}$ to replace it in the sequential optimization procedure. It is natural to estimate $F(z)$ in (\ref{111})by its empirical counterpart: $\hat{F}(z)=n^{-1}\sum_{i=1}^n \mathrm{I}(Z_i\leq z)$, as is done in \cite{Cui2015}. However, the empirical distribution is a step function which makes the optimization algorithm problematic. Instead, we use
\begin{equation*}
%\hat{F}_h(z)=n^{-1}\sum_{i=1}^n H\left(\frac{z-z_i}{h}\right),
\hat{F}(z):=\hat{F}_h(z)=\int_{-\infty}^{z} \hat{f}_{h}(u)\mathrm{d}u=n^{-1}\sum_{i=1}^n \int_{-\infty}^{z} K_{h}(Z_i-u)\mathrm{d}u
\end{equation*}
%where H can be any continuous cumulative distribution function (CDF). In this paper, we use Gaussian CDF $\Phi(\cdot)$ as our choice for $H$.
where $\hat{f}_h$ is a kernel density estimator of $f$ and $K_h(\cdot)=1/hK(\cdot/h)$ with $K$ being a kernel function and $h=h_n$ the bandwidth converging to 0 as $n\to\infty$. Similarly, $\hat{F}_r(z)$ can be estimated by
\[
\hat{F}_{r}(z):=\hat{F}_{hr}(z)= n_r^{-1}\sum_{j=1}^{n_r} \int_{-\infty}^{z} K_{h_r}(Z_j-u)\mathrm{d}u
\]
for $r=1, 2, \dots, R$, where $n_r$ is the sample size of the $r$th category and $Z_j$, $j=1, \dots, n_r$, are the sample points in this category. Then, we can use the estimator $\widehat{MV}$ given in $(\ref{mv})$ together with $\hat{F}_h(z)$ and $\hat{F}_{hr}(z)$ to implement the optimization algorithm. Note that $\bbeta^\T\bX$ is of one dimension, $\widehat{MV}(\bbeta):=\widehat{MV}(\bbeta^\T\bX|Y)$ can be estimated by the approach given above.
%Let $n_r$ be the sample size of the $r$th category and $X_{r_1},\cdots,X_{r_{n_r}}$ be the corresponding sample. An natural estimate $\widehat{MV}(\beta^TX|Y)$ can be written as

%\begin{equation}
%\sum_{r=1}^{R}\frac{n_r}{n}\sum_{i=1}^{n}\left[\frac{1}{n_r}\sum_{j=1}^{n_r}\Phi\left(\frac{\beta^TX_i-\beta^TX_{r_j}}{h_r}\right)- \frac{1}{n}\sum_{j=1}^n\Phi\left(\frac{\beta^TX_i-\beta^TX_j}{h}\right)\right]^2  \label{emv}
%\end{equation}

%In practise, we recommend the empirical bandwidth $h = k\cdot sd(U_i)n^{-1/3}$ where $k$ takes values around 3, $U_i=\hat{\beta}_0X$, $\hat{\beta}_0$ is some initial estimate of $\beta$, and $sd(U_i)$ is the sample standard deviation of $U_i$.

\setcounter{equation}{0}
\section{MMV in classification}
\label{s4}

\subsection{Fisher's LDA}

Consider a two-class classification problem. Suppose we have  $n$ labeled i.i.d. training samples $(Y_i,\bX_i), 1 \leq i \leq n$, where $\bX_i$ is a  $p$-dimensional feature vector and $Y_i  \in \{-1, 1\} $ is the corresponding class label. Let $p_1=\mathbb{P}(Y_i=1)$, $p_{-1}=\mathbb{P}(Y_i=-1)$ and assume
\begin{equation} \label{model}
\bX_i \sim N(Y_i \cdot \bmu, \; \bSigma),
\end{equation}
where $\bmu$ is the contrast mean vector between the two classes, and $\bSigma$ is the $p \times p$ covariance matrix. Given a new independent feature vector from the same population, i.e. $\bX \sim N(Y \cdot \bmu,  \,  \bSigma)$, our goal is to train  $(\bX_i, Y_i)$ to decide whether $Y = -1$ or $Y  = 1$. Here we use the contrast mean in model \eqref{model}, but the method and result below also applies to a more general model with mean vectors $\bmu_1, \bmu_2 \in \mathbb{R}^p$ with no extra difficulty.

Linear Discriminant Analysis (LDA), namely Fisher's LDA, is a well-known method for classification, which is essentially based on a weighted average of the test features $L(\bX) =  \sum_{j = 1}^p  w(j) X(j)$ and predicts  $Y =  \pm 1$ if $L(\bX) > <  0$. Here,  $\bw = (w_1, \ldots, w_p)^\T$ is a preselected weight vector. Fisher showed that the optimal weight vector satisfies
\begin{equation*} %\label{Fisherw}
\bw \propto \bOmega  \bmu.
\end{equation*}
where $\bOmega=\bSigma^{-1}$. In the classical setting  where  $n \gg p$, $\mu$ and $\Omega$ can be conveniently estimated and  Fisher's LDA is  approachable. Unfortunately, in the modern regime where $p \gg n$,  Fisher's LDA faces immediate challenges. To bypass the difficulty of estimating $\Omega$ in LDA, we propose a classifier based on the transformed variables $\bbeta_{0k}^\T \bX$,  $1 \leq k \leq d$, where $\bbeta_{0k}$s are the optimizers of the MMV procedure. The rationale of proposing this classifier is the intimate relationship given in the next theorem between the MMV method and the LDA.

\begin{theorem}\label{t6}
Under model (\ref{model}), $d=1$ and $\bbeta_{01} \propto \bOmega\bmu$, where $\bbeta_{01}$ is defined in (\ref{BetaPop}). % and $MV(\bbeta_{0i})=0$ for $i=2,3,\cdots,p$.
\end{theorem}

%\item Even in the simplest case  where $\Omega  = I_p$,    challenges remain, as the signals in $\mu$ can be weak. This is particularly true in micro-array data or SNP data.

%Let $\bB_0=(\bbeta_{01},\cdots,\bbeta_{0d})\in \mathbb{R}^{p \times d}$ and $\bnu=(\nu_1,\cdots,\nu_d)^\T$ be the trained LDA weight vector based on the transformed variables $\bbeta_{0k}^\T \bX$'s. Then $\bpi= {\bB_0}\bnu$ is the corresponding weight vector in the original scale. We conjecture that $d$ can be small.

In practice, if we have an estimator of $MV(Z|Y)$, say $\widehat{MV}(Z|Y)$ given in (\ref{mv}), then we can find the maximizer, denoted as $\hat\bbeta_1$, of $\widehat{MV}(Z|Y)$. For a new given feature vector $\bX$, we classify it as $Y=1$ if ${\hat\bbeta_1}^\top \bX>0$ and $Y=-1$ if ${\hat\bbeta_1}^\top \bX<0$. Since the $\hat\bbeta_1$ is an estimator of $\bbeta_{01}$, by Theorem 2 it is also an estimator of  $\bOmega\bmu$. This classifier has the following benefits.
%{\color{red}(I don't understand why in your original writeup we need to do LDA after the transformation ${\hat\bbeta_1}^\top \bX$. I removed the step of doing LDA in the writeup here since $\hat\bbeta_1$ is already an estimator of $\bbeta_{01}$. However I don't know how this will affect your simulation. In your simulation, if I understand correctly, you did LDA after the transformation, right? I just feel, if you do LDA after transformation you need to justify it.)}

\begin{enumerate}

\item It does not have to estimate the precision matrix $\bOmega$ and can be implemented efficiently in the case where $n < p$.

\item It requires minimal conditions. Given $Y$, the transformed variables ${\bbeta}_{01}^\T\bX, \cdots,{\bbeta}_{0d}^\T\bX$ have distributions closer to normal than the original variable $\bX$ when $d\ll p$ (\cite{HallLi1993}). It is well known that the conditional normality is an important assumption for linear discriminant analysis.

\item It has better performance in classification than the traditional LDA. We will show in our simulation study that it reduces the classification error significantly compared with the LDA.

%\item We can use the estimated first direction to do classification directly. Simulation shows that $\hat\bbeta_1$ almost equals the weight vector giving by LDA. Then the other estimated directions would add efficiency if we do MMV first then LDA. Moreover, when $p>n$, classification via MMV will avoid estimating assuming a sparsity structure for the precision matrix which seems to be somewhat subjective.
\end{enumerate}

This theorem states that for model (\ref{model}), the true $d=1$ and $\bbeta_{01}^\T \bX$ contains all the information for classification. This means at the population level (assuming population functions are all known), the MMV procedure gives exactly the LDA classifier. The theorem also justifies the efficiency of the LDA for the normal model (\ref{model}) in terms of maximum mean variance. When $p>n$, the LDA needs to estimate the inverse of the covariance matrix; thus it is unsolvable or requires extra sparsity assumption. The MMV is an efficient alternative to circumvent this problem. %Note that when the distributions are known in the LDA context, we can compute the MV index by the parametric method: $\bbeta_{01} = \bSigma^{-1}\bmu/\|\bSigma^{-1}\bmu\|$. It's exactly the estimated optimal weight of the LDA. That's not surprising, because the LDA optimal weight is a MLE estimator under this setting.

In the field of high dimensional classification, we are usually faced with rare and weak signals, i.e. important features are sparse and each contributes weakly to the classification decision. In this setting, reasonable classification becomes quite difficult. If rare and weak signals imply sparse optimal weights, the proposed MMV estimator can be adjusted by adding a penalty term to the objective function or by some kinds of thresholding. We conjecture that higher criticism thresholding (HCT) might also be applied in this context. It has been verified that HCT performs quite well theoretically and practically when the signals are rare and weak; see \cite{DonohoJin2008} and \cite{Fan2013} for details. Hence, the extension of HCT to the MMV estimator seems rather direct and reasonable. Nevertheless, a theoretical investigation of this would be very complicated, and we leave it for further research.

\subsection{Index model}

The index model enjoys a lot of popularity in regression and classification. The logistic model and probit model are special cases where the link function is known with a single index. A general index model can be expressed as the following semi-parametric model. Let $Y \in \mathbb{R}$ denote the response variable and $\bX \in \mathbb{R}^p$ denote the covariates. Assume there exist orthogonal $p$-dimensional vectors $\bbeta_1, \dots, \bbeta_k$ with unit norm such that
\be \label{e11}
Y=f(\bbeta_1^\T\bX,\dots,\bbeta_k^\T\bX, \varepsilon) \quad (k<p),
\ee
where $f$ is an arbitrary unknown link function and $\varepsilon$ is independent of $\bX$. With a bit of an abuse of the notation, the notation $k$ in (\ref{e11}) can be seen as a fixed integer indicating the number of the indexes. The column space spanned by $\{\bbeta_1, \dots, \bbeta_k\}$ is defined as the central dimension reduction subspace by \cite{Cook1994} and \cite{Cook1998}. Under the setting of index model (\ref{e11}), we can detail Theorem \ref{t1} to some extent. Let $\bB=(\bbeta_1,\dots,\bbeta_k)^\T$. Assume
\be \label{e8}
Y \independent \bX|\bB \bX
\ee
and there exists a $p$-dimensional vector $\bgamma$ such that
\be \label{e9}
\bB\bgamma=\mathbf{0},\quad\bgamma^\T \bX \independent \bB\bX,
\ee
then by Lemma 4.3 in \cite{Dawid1979} and Proposition 4.6 in \cite{Cook1998}, we can get $MV(\bgamma^\T \bX|Y)=0$. This implies that under mild conditions, the MMV method can recover all the information in $\bX$ related to classification with $d<p$ indexes in Theorem \ref{t1}. Specifically, when $\bX \sim N(\bmu,\bSigma)$, we can obtain the following theorem.

\begin{theorem} \label{t5}
Assume $\bX \sim N(\bmu,\bSigma)$ and $Y \independent \bX|\bB \bX$ where $\bB=(\bbeta_1,\dots,\bbeta_k)^\T$ with $k<p$. For $d$ defined in Theorem \ref{t1}, if $2k\le p$, then $d \le 2k$. Specifically, if $\bSigma=\bI$, then $d=k$.
\end{theorem}

Theorems \ref{t1} and \ref{t5} indicate that when $\bSigma=\bI$, the MMV procedure can exactly recover the central subspace with $d=k$ steps under the setting of index model and normal covariates. To be specific, in the logistic (probit) model with normal covariates, $\bbeta_{01}\propto\bbeta_{1}$, where $\bbeta_1$ denotes the coefficient vector of the logistic (probit) model and $\bbeta_{01}$ is defined in Theorem \ref{t1}. This implies that $d=1$ is enough for the logistic (probit) model. The advantage of our method is that it is a nonparametric method, thus we do not need to assume any specific link functions.

\subsection{Other classification algorithms}

Other popular classification methods such as K-Nearest Neighbours (KNN), neural networks and support vector machine (SVM), can be connected to MMV by a two step procedure, i.e. dimension reduction first and classification next. We conjecture that such a two step procedure will improve the accuracy of the classification because high dimensionality causes problems in the classification algorithms mentioned above. Simulations in Section \ref{s5} validate this conjecture.

\setcounter{equation}{0}
\section{Consistency and asymptotic normality}
\label{s3}

We introduce the following notations to simplify the description of the theory. %The responsible variable $Y$ is categorical with $R$ classes $\{y_1, \dots, y_R\}$ and $p_r = \mathbb{P}(Y=y_r)$. The covariate $\bX \in \mathbb{R}^p$ is a continuous vector.
$\{\bX_i,Y_i\}_{i=1}^n$ are i.i.d. samples. $n_r$ indicates the number of samples in the class $Y=y_r$ for $r=1,\dots,R$. %For any continuous variable $X$ and its i.i.d. samples $\{X_i\}_{i=1}^n$, denote $\hat{F}_{h}(x)=n^{-1}\sum_{i=1}^n \int_{-\infty}^{x}K_h(X_i-t) dt$, $\hat{F}_{hr}(x)= 1/n_r \sum_{j=1}^{n_r}\int_{-\infty}^{x}K_h(X_j-t) dt$ and $\hat{F}(x)=n^{-1}\sum_{i=1}^n I\{X_i\le x\}$.
For the case $d=2$ where $d$ is defined in Theorem \ref{t1}, denote
\[
L^{(1)}(\btheta_1) =: L^{(1)}(\bbeta_1, \lambda_1) = MV(\bbeta_1^{\T}\bX|Y) + \lambda_1 (\bbeta_1^{\T}\bbeta_1-1),
\]
\[
L_{nh}^{(1)}(\btheta_1) =: L_{nh}^{(1)}(\bbeta_1, \lambda_1) = MV_n(\bbeta_1^{\T}\bX|Y) + \lambda_1 (\bbeta_1^{\T}\bbeta_1-1),
\]
and
\[
L^{(2)}(\btheta_2) =: L^{(2)}(\bbeta_2, \blambda_2) = MV(\bbeta_2^{\T}\bX|Y) + \lambda_{21} (\bbeta_2^{\T}\bbeta_2-1) + \lambda_{22}(\bbeta_{01}^{\T}\bbeta_2),
\]
\[
L_{nh}^{(2)}(\btheta_2) =: L_{nh}^{(2)}(\bbeta_2, \blambda_2) = MV_n(\bbeta_2^{\T}\bX|Y) + \lambda_{21} (\bbeta_2^{\T}\bbeta_2-1) + \lambda_{22}(\widehat\bbeta_{1}^{\T}\bbeta_2),
\]
where $\blambda_2 =(\lambda_{21},\lambda_{22})^{\T}$. Let $\btheta_{0i}=(\bbeta_{0i}^\T,\blambda_{0i})^\T$ and $\hat\btheta_i=(\hat\bbeta_i^\T,\hat\blambda_i)^\T$ be the maximizers of $L^{(i)}(\btheta_i)$ and $L_{nh}^{(i)}(\btheta_i)$ respectively for $i=1,2$. Let $\Omega_i$ denote the parameter space of $\bbeta_i$ and $B(\kappa_i)=\{\bbeta_i: \|\bbeta_i-\bbeta_{0i}\|\le\kappa_i\}$ be a ball with center $\bbeta_{0i}$ and radius $\kappa_i$ for $i=1,2$. The boundary of the ball is denoted by $\partial B(\kappa_i)$. Denote $\Gamma_i = \{\bbeta_i: \bbeta_i^{\T}\bbeta_i=1\}$ for $i=1,2$ and $\Upsilon=\{\bbeta_2: \bbeta_{01}^{\T}\bbeta_2=0\}$. Let $C(\kappa_i)=\{\bbeta_i: \|\bbeta_i-\bbeta_{0i}\|\le\kappa_i\}$ be a complex ball in $\mathbb{C}^p$ with center $\bbeta_{0i}$ and radius $\kappa_i$ for $i=1,2$. Note that each element of $\bbeta_i$ is complex. Denote $\Gamma_{Ci} = \{\bbeta_i: \bbeta_i^{\T}\bbeta_i=1\}$ for $i=1,2$ and $\Upsilon_C=\{\bbeta_2: \bbeta_{01}^{\T}\bbeta_2=0\}$. For $i=1,2$, let $\Omega_{Ci}$ be the parameter space of the complex $\bbeta_{i}$.

The conditions below are required to establish the consistency and asymptotic normality.

$(1)$ $c_1/R \le \min_{1 \le r \le R} p_r \le \max_{1 \le r \le R} p_r \le c_2/R$ and $R=O(n^{\delta})$ with $0<\delta \le 1/2$.

$(2)$ a. There exists an open subset $\omega_1$ of $\Omega_1 \cap \Gamma_1$ that contains the true parameter $\bbeta_{01}$ for almost all $(\bX, Y)$ and $\sup_{\bbeta_1 \in B(\kappa_{01})} MV(\bbeta_1^{\T}\bX|Y)<\infty$, and for any $\kappa_1 \in (0, \kappa_{01}]$
\[
\sup_{\bbeta_1\in\partial B(\kappa_1)\cap \Gamma_1} MV(\bbeta_1^{\T}\bX|Y)<MV(\bbeta_{01}^{\T}\bX|Y).
\]
b. There exists an open subset $\omega_2$ of $\Omega_2 \cap \Gamma_2 \cap \Upsilon$ that contains the true parameter $\bbeta_{02}$ for almost all $(\bX, Y)$ and $\sup_{\bbeta_2 \in B(\kappa_{02})} MV(\bbeta_2^{\T}\bX|Y)<\infty$, and for any $\kappa_2 \in (0, \kappa_{02}]$
\[
\sup_{\bbeta_2\in\partial B(\kappa_2)\cap \Gamma_2 \cap \Upsilon} MV(\bbeta_2^{\T}\bX|Y)<MV(\bbeta_{02}^{\T}\bX|Y).
\]

$(3)$ $\int uK(u) du=0$, $\int u^2 K(u) du<\infty$ and $nh^4\to 0$ where $h=h_1=\dots=h_R$; $|K|_{\infty}=\sup_{u \in \mathbb{R}}|1/hK(u/h)|<\infty$.

$(4)$ There exists a constant $\kappa_{01}$ such that $MV(\bbeta_1^{\T}\bX|Y)$ is an analytic function of each coordinate of $\bbeta_1$ in $C(\kappa_{01}) \subseteq \Omega_{C1} \cap \Gamma_{C1}$ and $\sup_{\bbeta_1 \in C(\kappa_{01})} MV(\bbeta_1^{\T}\bX|Y)<\infty$. For any $\bbeta_1 \in C(\kappa_{01})$, $MV^{'}(\bbeta_1)=MV^{'}(\bbeta_1^{\T}\bX|Y)$ and $MV^{''}(\bbeta_1)=MV^{''}(\bbeta_1^{\T}\bX|Y)$ exist. $\sup_{\bbeta_1 \in C(\kappa_{01})}\|MV^{'}(\bbeta_1)\|_{\infty} < \infty$ and $\sup_{\bbeta_1 \in C(\kappa_{01})}\|MV^{''}(\bbeta_1)\|_{\infty} < \infty$, where $\|MV^{'}(\bbeta_1)\|_{\infty}=\inf\{C>0:\mathbb{P}(\|MV^{'}(\bbeta_1)\|\le C)=1\}$ and $\|MV^{''}(\bbeta_1)\|_{\infty}=\inf\{C>0:\mathbb{P}(\|MV^{''}(\bbeta_1)\|\le C)=1\}$. (Note that we only give the condition for $i=1$. The conditions for the $i=2$ case are a bit more complex but quite similar to the $i=1$ case.)

$(5)$ $L^{''}(\btheta_0)$ is nonsingular and $\E (\|\bX/\sqrt{p}\|)< \infty$.

Condition $(1)$ requires the proportion of each response class to be moderate, not too small nor too large. $R=O(n^{\delta})$ allows the number of the classes to grow with the sample size. %and the subscript $n$ of $R_n$ is used to emphasize this point.
This condition is also imposed by \cite{Cui2015}. Condition $(2)$ is assumed to ensure the existence of the MMV optimizers. A similar condition is assumed in \cite{Chen2017} for a likelihood function. Simulation results below also verify this assumption. Condition $(3)$ is widely used in the literature of kernel density estimators and is assumed for the consistency of the estimator; see \cite{Cheng2017} for reference. Condition $(4)$ is not as strict as it seems. Recall that the MV index is defined on the cumulative distribution functions which are, of course, bounded. Therefore, we only need the bounded constraints for the corresponding density function and the first derivative of the density function, which are both quite mild. Condition $(5)$ is required to guarantee the root-$n$ consistency of the proposed estimator; it is in the spirit of the Von Mises proposition (\cite{Serfling1980}, Section $6.1$).

\begin{theorem}[Consistency]\label{t2}
Let $\{\bX_i, Y_i\}_{i=1}^n$ be i.i.d. samples. For both fixed and diverging $R$, under Conditions $(1)$ -- $(3)$, we have $\widehat\bbeta_i \to \bbeta_{0i}$ in probability as $n \to \infty$ for $i =1, 2$. Moreover, when $p$ satisfies $p^{p/2} n^{-\alpha(1-\delta)}=o(1)$ for any $\alpha<1/2$, under Conditions $(1)$ -- $(3)$, $\widehat\bbeta_i \to \bbeta_{0i}$ in probability as $n \to \infty$ for $i =1, 2$.
\end{theorem}

Theorem \ref{t2} indicates that the MMV estimators are consistent for both fixed and diverging $p$ and $R$ cases under regular conditions. Here, we only present the theorem for the case $d=2$ for the convenience of statement. When $p$ is fixed, the $\sqrt{n}$ consistency and asymptotic normality can be further proved.

\begin{theorem}[Asymptotic normality]\label{t3}
Let $\{\bX_i, Y_i\}_{i=1}^n$ be i.i.d. samples. Under Conditions $(1)$ -- $(5)$, %if $p$ is fixed,
$\sqrt{n}(\widehat\bbeta_i-\bbeta_{0i})$ is asymptotically normal with mean zero and covariance matrix $V_i$, where $V_i$ is defined in the proof for $i=1,2$.
\end{theorem}

When $p$ is diverging with $n$, the asymptotic normality given above may or may not hold, which depends on the property of the MV index, thus the distribution of $\bX$ and $\bX|Y$. When the uniform convergence stays for the first and second derivatives of the estimated MV index with a properly fast convergence rate, the convergence rate of the MMV estimator can be proved via techniques similar to those employed in \cite{FanPeng2004} for the non-concave penalized likelihood when the number of parameters is diverging with the sample size. Consider the simple case $d=1$. More conditions are needed (the subscript ``1'' of $\bbeta_{01}$, $\btheta_1$ and $\btheta_{01}$ and the superscript ``(1)'' of $L^{(1)}$ and $L^{(1)}_{nh}$ are omitted for simplicity):

$(6)$ $\|\nabla MV_n(\bbeta_0)-\nabla MV(\bbeta_0)\|=O_p(\alpha_n)$ where $\alpha_n=\alpha_n(n,p,h)$ and $p^{3/2}\alpha_n =o_p(1)$.

$(7)$ $\lambda_{\max} \{\nabla^2 MV_n(\bbeta_0)-\nabla^2 MV(\bbeta_0)\}=o_p(1)$.

$(8)$ There is a large enough ball centering at $\btheta_{0}$ such that for any $\btheta$ in the ball,
\[
|\frac{\partial L(\btheta)}{\partial\theta_i \partial\theta_j \partial\theta_k}|\le M_1 <\infty \qquad |\frac{\partial L(\btheta)}{\partial\theta_i \partial\theta_j \partial\theta_k}-\frac{\partial L_{nh}(\btheta)}{\partial\theta_i \partial\theta_j \partial\theta_k}| \le M_2 <\infty
\]
for $1 \le j,k,l \le p$ where $M_1$ and $M_2$ are positive numbers.

\begin{theorem}\label{t4}
If Conditions $(1)$-$(8)$ hold, there exists a local maximizer $\hat\bbeta$ such that $\|\hat\bbeta-\bbeta_0\|=O_p(\alpha_n)$ where $\alpha_n$ is defined in Condition $(6)$.
\end{theorem}

\setcounter{equation}{0}
\section{Simulations}
\label{s5}

 Although the MMV method can be readily used under settings of $n < p$, the computation cost for a sequential algorithm like ours is quite high. It is usually better to use cross validation to choose the proper number of indexes $d$ and empirical bandwidth $h$ in practice. However, the computation cost of doing so is also very high. For simplicity, instead, in the following simulation studies, we use $d$ as the true dimension of the central subspace and $h = 3\cdot sd(\tilde{\bbeta}_1)n^{-1/3}$ where $sd$ stands for standard deviation and $\tilde{\bbeta}_1$ is a good initial estimate of $\bbeta_1$. If the predictors' dimension is ultra-high, our suggestion is to conduct feature screening first to reduce dimensionality $p$ (say, $\exp(O(n^\xi))$ for some $\xi >0$) to a relatively large scale $d'$ (\eg, o(n)) by fast methods such as those of \cite{FanFan2008} and \cite{Cui2015}. When the size of $p$ is comparable to $n$, our method is quite fast and effective. We use ten-fold cross validation to calculate the classification error in both simulation and real data analysis. We repeat the experiment 400 times and the average classification error and the corresponding standard deviation (in parentheses) are then calculated. Let $\bbeta_1=(1,1,1,1,0,\ldots,0)^\T$ and $\bbeta_2=(1,-1,1,-1,0,\ldots,0)^{\T}$. The calculation for LDA, logistics regression, SVM and KKN is based on the corresponding Matlab (R2015a) packages using default settings.

\subsection{Fisher's linear discriminant analysis }
In this study, we set $p$ be 50 and 200 respectively, with the sample size $n=80$. We generate $\bY=(1,\ldots,1,-1,\ldots,-1)^{\sf T}$ first, then generate $\bX$ as follows. It is an ordinary LDA model which is in fact an inverse model with a one-dimensional central subspace.

\vspace{0.5cm} \textsc{Model I}
\begin{equation*}
\bX= \bbeta_1 Y + \bDelta\bepsilonbb,
\end{equation*} where $\bepsilonbb \sim N({\mathbf 0}, {\bf I}_{n})$ and $\bDelta_{ij} = 0.5^{|i-j|}$ for $1 \leq i,j \leq
10$.

    \begin{table}[ht]
         \centering
         \caption{Average classification error (percentage)}
         \vspace{0.1cm}
    \renewcommand{\arraystretch}{1.1}
   \tabcolsep 5pt
       \begin{tabular}{cccccccccc}
   \hline  Method & \multicolumn{1}{c}{MMV+LDA} &  &
   \multicolumn{1}{c}{LDA} & &
       \\ \cline{2-2} \cline{4-4}

                   &                  \multicolumn{3}{c}{$p=50$}                \\
     Model I              & 9.95 (3.43)  &       & 24.15 (6.21)         \\
            &                  \multicolumn{3}{c}{$p=200$}                \\
     Model I              & 13.87 (4.42)  &       & 19.83 (6.33)         \\

     \\\hline
        \end{tabular}
    \end{table}

Table 1 shows that MMV+LDA outperforms LDA significantly in both settings ($p=50$ and $p=200$). By applying MMV, the classification error is decreased by 50 percent or so. Although MMV+LDA equals LDA in the population level, the former does work better in the finite sample settings. The reason is that dimension reduction through MMV increases estimation efficiency.

\subsection{Logistic regression}
In this study, we set $p$ be 20 and 50 respectively, with the sample size $n=80$.  Since logistic regression utilizes likelihood estimation, the sample size is required to be larger than the dimension of the predictors. We generate data using the logistic model as follows.

\vspace{0.5cm} \textsc{Model II}
\begin{equation*}
Y = \mathrm{I}\left(1/\{1+\exp(\bbeta_1^{\T}\bX)\} \geq 0.5\right),
\end{equation*} where $\mathrm{I}$ is the indictor function and  $\bX = (X_1,\ldots,X_{p})^{\ T} \sim N({\bf 0},\bPsi)$ with
$\bPsi_{ij} = 0.5^{|i-j|}$ for $1 \leq i,j \leq p $. In this study, the central subspace is
spanned by the direction $\bbeta_1$.

    \begin{table}[ht]
         \centering
         \caption{Average classification error (percentage)}
         \vspace{0.1cm}
    \renewcommand{\arraystretch}{1.1}
   \tabcolsep 5pt
       \begin{tabular}{cccccccccc}
   \hline  Method & \multicolumn{1}{c}{MMV+Logistic regression} &  &
   \multicolumn{1}{c}{Logistic regression} & &
       \\ \cline{2-2} \cline{4-4}

                   &                  \multicolumn{3}{c}{$p=20$}                \\
     Model II              & 9.33 (3.36)  &       & 13.31 (4.16)         \\
            &                  \multicolumn{3}{c}{$p=50$}                \\
     Model II              & 14.85 (4.36)  &       & 31.71 (6.27)         \\

     \\\hline
        \end{tabular}
    \end{table}

It can been seen clearly from Table 2 that MMV, as a dimension reduction technique, improves estimation efficiency, and thus reduces classification error remarkably when it is combined with Logistic regression.

\subsection{More complex models}
We compare our method with more advanced algorithms like SVM and KKN in this study. Models III and IV are multiple index models. We set $p$ be 50 and 200 respectively, while the sample size $n=160$. In these two models, the central subspace is spanned by the directions $\bbeta_1$ and $\bbeta_2$.

\vspace{0.5cm} \textsc{Model III}
\begin{equation*}
Y = \mathrm{I}\left(\bbeta_1^{\T}X/\{0.5 + (\bbeta_2^{\T}X + 1.5)^2\} + 0.2\epsilon \geq 0\right),
\end{equation*} where $\epsilon \sim N(0,1)$, $\bX = (X_1,\ldots,X_{p})^{\T} \sim N({\bf0},\bPsi)$ with
$\bPsi_{ij} = 0.5^{|i-j|}$ for $1 \leq i,j \leq p $, and $\bX
\independent \epsilon$.

\vspace{0.5cm} \textsc{Model IV}
\begin{equation*}
Y = \mathrm{I}\left((\bbeta_1^{\T}\bX)^2  + (\bbeta_2^{\T}\bX)^2 + 0.2\epsilon \geq 1\right),
\end{equation*} where $\epsilon \sim N(0,1)$, $\bX = (X_1,\ldots,X_{p})^{\sf T} \sim N({\bf0},\bPsi)$ with
$\bPsi_{ij} = 0.5^{|i-j|}$ for $1 \leq i,j \leq p $, and $\bX
\independent \epsilon$.

    \begin{table}[ht]
         \centering
         \caption{Average classification error (percentage)}
         \vspace{0.1cm}
    \renewcommand{\arraystretch}{1.1}
   \tabcolsep 5pt
       \begin{tabular}{cccccccccc}
   \hline  Method & \multicolumn{1}{c}{MMV+SVM} & & \multicolumn{1}{c}{SVM}  & & \multicolumn{1}{c}{MMV+KKN} & & \multicolumn{1}{c}{KKN} & &
       \\ \cline{2-2} \cline{4-4} \cline{6-6} \cline{8-8}

                   &                  \multicolumn{7}{c}{$p=50$}                \\
     Model III              & 16.65 (3.61)  &       & 20.49 (4.01)&     & 17.80 (3.75)&  & 34.52 (4.51)     \\
     Model IV              & 16.66 (4.09)  &       & 25.22 (4.46) &    & 18.33 (4.65)&  & 22.75 (5.43)     \\
            &                  \multicolumn{7}{c}{$p=200$}                \\
     Model III              & 25.18 (4.05)  &       & 27.08 (4.45)&     & 25.63 (3.99)&  & 42.21 (4.40)     \\
     Model IV              & 16.20 (5.02)  &       & 19.35 (5.89)&     & 14.09 (4.06)&  & 22.70 (6.45)     \\

     \\\hline
        \end{tabular}
    \end{table}

Table 3 indicates that even for complex classification techniques such as SVM and KNN, employing MMV before classification still enjoys a significant decrease of the classification error. This further conforms the efficiency and priority of our proposed method.

\subsection{Real data analysis}
We apply our method to human colon cancer data with $n=62$ and $p=2000$, which is available in R. There are 40 samples from tumors ``t'', and 22 samples are from normal ``n'' biopsies. The data was originally collected on microarrays with 6500 probes. 2000 of them were selected apparently, randomly, to be used for demonstrating statistical methods. We first screen the number of the predictors to 100 by the method of \cite{Cui2015}, and compare our methods with LDA, SVM and KKN. We apply the same bandwidth selection and cross validation methods as those in the simulation study. A few choices of the dimension $d=1,2,3$ are tried and they have rather similar results. Here we only present the result with fixed $d=1$. To get a fair comparison, we repeat the permutation 100 times for cross validation results. The table below summarizes the average classification errors and the corresponding standard deviations.

    \begin{table}[ht]
         \centering
         \caption{Average classification error (percentage)}
         \vspace{0.1cm}
    \renewcommand{\arraystretch}{1.1}
   \tabcolsep 5pt
       \begin{tabular}{cccccccccc}
   \hline  & \multicolumn{1}{c}{MMV+LDA} & & \multicolumn{1}{c}{LDA}  \\
   & 11.24 (1.16)  &       & 16.74 (3.00) \\
      \hline  & \multicolumn{1}{c}{MMV+SVM} & & \multicolumn{1}{c}{SVM}  \\
   & 12.40 (1.52)  &       & 16.53 (2.41) \\
   \hline & \multicolumn{1}{c}{MMV+KKN} & & \multicolumn{1}{c}{KKN} \\
   & 14.37 (2.14)  &       & 18.73 (1.43)
     \\\hline
        \end{tabular}
    \end{table}

Table 4 demonstrates that ``MMV+.'' performs much better than the original classification method. It seems that the performances of the three methods are comparable to each other, with SVM performing a little better, while MMV+LDA performs the best among the three ``MMV+.'' methods. We conjecture that the relationship between MMV and the original classification techniques may be more than simple addition.

\setcounter{equation}{0}
\section{Conclusion}
\label{s6}

In this paper, we propose a new sufficient dimension reduction approach - Maximal Mean Variance - which is designed deliberately for high-dimensional classification. Our method requires fairly mild restrictions on the predicting variables and keeps the model-free advantage without the need to estimate the link function; thus it can be widely used. The relationships between MMV and popular classification methods are discussed in detail. The consistency and convergence rate of the MMV estimator are established for both fixed and diverging dimension ($p$) cases and the number of the response classes can also be allowed to grow with the sample size $n$. When the dimension of the predicting vector is fixed, we further construct the asymptotic normality for the MMV estimator. Simulations and real data analysis validate the efficiency and priority of the proposed method. While cross validation can be employed to choose a proper $d$ in practice, it would be quite challenging to derive an optimal $d$ theoretically under some kind of criterion because the algorithm of MMV is a stepwise procedure. We leave it for further research. Besides, the MMV method can be readily applied to the ultra high dimensional setting by conducting a screening procedure first. This two-scale learning framework is in the sprit of \cite{FanLv2008} for sure independence screening.

\setcounter{equation}{0}
\section{Appendix}
\label{s7}

\subsection{Proof of Theorem \ref{t1}}

\begin{proof}
By the property of MV index, we have $\bbeta_{0i}^{\T}\bX \not\independent Y$ for $i=1,\dots,d$ and $\bbeta_{0i}^{\T}\bX \independent Y$ for $i=d+1,\dots,p$. Recalling the definition of central subspace, this argument means $\bbeta_{0i}^{\T} \notin S_{Y|\bX}^{\perp}$ for $i=1,\dots,d$ and $\bbeta_{0i}^{\T} \in S_{Y|\bX}^{\perp}$ for $i=d+1,\dots,p$ where $S_{Y|\bX}^{\perp}$ is the orthogonal complement space of $S_{Y|\bX}$. Noting that $\bbeta_{0i}^{\T}\bbeta_{0i}=1$, $\bbeta_{0i}^{\T}\bbeta_{j0}=0$ for $i,j=1,\dots,p$ and $i \ne j$ and ${\rm span}\{\bbeta_{(d+1)0},\dots,\bbeta_{p0}\}\subseteq S_{Y|\bX}^{\perp}$, we easily obtain $S_{Y|\bX}\subseteq {\rm span}\{\bbeta_{01},\dots,\bbeta_{0d}\}$.

For any $k<d$, if ${\rm span}\{\bbeta_{01},\dots,\bbeta_{0k}\} \supseteq S_{Y|\bX}$, we can obtain for any $i> k$, $\bbeta_{i}^{\T}\bX\independent Y$ which leads to a contradiction at least for $i=d>k$.

\end{proof}

\subsection{Proof of Theorem \ref{t6}}

\begin{proof}

Let $Z$ denote the variable $\bbeta^\T \bX$ and $Z_i$ denote the variable $\bbeta^\T\bX|Y=i$ ($i=1,-1$) for $\bbeta \in \mathbb{R}^p$ satisfying $\|\bbeta\|=1$. Then, $Z_1\sim N(\bbeta^\T \bmu, \bbeta^\T\bSigma\bbeta)$, $Z_{-1}\sim N(-\bbeta^\T \bmu, \bbeta^\T\bSigma\bbeta)$ and
\[
\begin{split}
F(Z)&=p_1 \cdot F_1(Z)+p_{-1} \cdot F_{-1}(Z)\\
&=p_1\cdot \Phi((Z-\bbeta^\T \bmu)/\sqrt{\bbeta^\T\bSigma\bbeta})+p_{-1}\cdot \Phi((Z+\bbeta^\T \bmu)/\sqrt{\bbeta^\T\bSigma\bbeta}),
\end{split}
\]
where $F_i(Z)$ denotes the conditional distribution function of $Z|Y=i$ for $i=1,-1$.
Then,
\[
\begin{split}
MV(\bbeta)& \ = \ p_1\int [F_1(z)-F(z)]^2 \mathrm{d}F(z)+p_{-1}\int [F_{-1}(z)-F(z)]^2 \mathrm{d}F(z)\\
%&=p_1 p_{-1}^2 \int [F_1(z)-F_{-1}(z)]^2\mathrm{d}F(z) + p_{-1}p_1^2 \int [F_1(z)-F_{-1}(z)]^2\mathrm{d}F(z)\\
& \ = \ \left(p_1 p_{-1}^2+p_{-1}p_1^2\right) \left\{ p_1 \int [F_1(z)-F_{-1}(z)]^2 \mathrm{d}F_1(z)+p_{-1} \int [F_1(z)-F_{-1}(z)]^2 \mathrm{d}F_{-1}(z) \right \}\\
%&=p_1 p_{-1} \{ p_1 \int [\Phi((z-\bbeta^\T \bmu)/\sqrt{\bbeta^\T\bSigma\bbeta})-\Phi((z+\bbeta^\T \bmu)/\sqrt{\bbeta^\T\bSigma\bbeta})]^2 \mathrm{d}\Phi((z-\bbeta^\T \bmu)/\sqrt{\bbeta^\T\bSigma\bbeta})\\
%&~~+p_{-1} \int [\Phi((z-\bbeta^\T \bmu)/\sqrt{\bbeta^\T\bSigma\bbeta})-\Phi((z+\bbeta^\T \bmu)/\sqrt{\bbeta^\T\bSigma\bbeta})]^2 \mathrm{d}\Phi((z+\bbeta^\T \bmu)/\sqrt{\bbeta^\T\bSigma\bbeta})\}\\
& \ = \ p_1 p_{-1} \left\{ p_1 \int \left[\Phi(t)-\Phi(t+2\bbeta^\T \bmu/\sqrt{\bbeta^\T\bSigma\bbeta})\right]^2 \mathrm{d}\Phi(t) \right.\\
&\ \left.~~~~+p_{-1} \int \left [\Phi(t)-\Phi(t-2\bbeta^\T \bmu/\sqrt{\bbeta^\T\bSigma\bbeta})\right]^2 \mathrm{d}\Phi(t) \right \}.\\
\end{split}
\]
\\
Hence, we only need to maximize $\bbeta^\T \bmu/\sqrt{\bbeta^\T\bSigma\bbeta}$. The solution is exactly the optimal weight of LDA, i.e. $\bbeta_{01} \propto \bSigma^{-1}\bmu$, where $\bbeta_{01}$ is defined in Theorem \ref{t1}. %Refer to \\\verb|https://en.wikipedia.org/wiki/Linear_discriminant_analysis| (Section: Fisher's linear discriminant) for details.
Thus, if we choose $d=1$ in the MMV procedure, then MMV+LDA equals LDA at the population level.

Then we try to find $\bbeta_{02}$. We always assume $\bSigma$ is positive definite, and thus for any non-zero $\bbeta$, $\bbeta^\T \bSigma\bbeta>0$. Since now we maximize $(\bbeta^\T  \bSigma\bbeta)^{-1/2}\bbeta^\T  \bmu$ subject to $\bbeta^\T \bbeta=1$ and $\bbeta^\T \bSigma^{-1}\bmu=0$, let
$$
L \ = \ (\bbeta^\T  \bSigma\bbeta)^{-1/2}\bbeta^\T  \bmu+\lambda(\bbeta^\T  \bbeta-1)+\pi \bbeta^\T \bSigma^{-1}\bmu.
$$
Take first partial derivative w.r.t. $\bbeta$ and set it to zero we have
$$
\frac{\partial L}{\partial \bbeta} \ = \ (\bbeta^\T  \bSigma\bbeta)^{-1/2}\bmu^\T -\bbeta^\T  \bmu (\bbeta^\T  \bSigma\bbeta)^{-3/2}\bbeta^\T \bSigma+2\lambda\bbeta^\T +\pi\bmu^\T  \bSigma^{-1} \ = \ 0.
$$
Then
\[
\begin{split}
0 \ &  = \ \frac{\partial L}{\partial \bbeta} \bbeta \\
& = \ (\bbeta^\T  \bSigma\bbeta)^{-1/2}\bmu^\T \bbeta-(\bbeta^\T  \bSigma\bbeta)^{-1/2} \bbeta^\T  \bmu +2\lambda\bbeta^\T \bbeta +\pi\bmu^\T  \bSigma^{-1}\beta \\
& = \ 2\lambda,
\end{split}
\]
i.e. $\lambda=0$. Now $(\bbeta^\T  \bSigma\bbeta)^{3/2}\frac{\partial L}{\partial \bbeta}=0$ is reduced to
\begin{equation}\label{1}
0 \ = \ (\bbeta^\T \bSigma\bbeta)\bmu^\T -\bbeta^\T \bmu\bbeta^\T  \bSigma+(\bbeta^\T  \bSigma\bbeta)^{3/2} \pi \bmu^\T  \bSigma^{-1}. %\\
%& = \ \bbeta^\T  \left [ \bSigma\bbeta\bmu^\T  - \bmu\bbeta^\T \bSigma +\bSigma \bbeta \bmu^\T  \pi (\bbeta^\T  \bSigma\bbeta)^{1/2} \bSigma^{-1} \right] \\
%& = \ \bbeta^\T  \left [ U-U^\T +W \right], {\rm say}.
\end{equation}
Multiply both sides of above equation by $\bSigma^{-1}\bbeta$ from right, together with $\bbeta^\T \bbeta=1$ and $\bbeta^\T \bSigma^{-1}\bmu=0$, we have
\begin{equation}\label{2}
\begin{array}{lll}
0 & = & \displaystyle (\bbeta^\T \bSigma\bbeta)\bmu^\T \bSigma^{-1}\bbeta-\bbeta^\T \bmu\bbeta^\T  \bSigma\bSigma^{-1}\bbeta+(\bbeta^\T  \bSigma\bbeta)^{3/2}\pi \bmu^\T  \bSigma^{-1}\bSigma^{-1}\bbeta \\
& = & \displaystyle0- \bbeta^\T \bmu+\pi(\bbeta^\T  \bSigma\bbeta)^{3/2} \bmu^\T \bSigma^{-2}\bbeta \\
& = & \displaystyle \pi(\bbeta^\T  \bSigma\bbeta)^{3/2} \bmu^\T \bSigma^{-2}\bbeta-\bmu^\T  \bbeta,
\end{array}
\end{equation}
and thus
\begin{equation}\label{3}
\pi=(\bbeta^\T  \bSigma\bbeta)^{-3/2} \left(\bmu^\T \bSigma^{-2} \bbeta\right)^{-1}\left(\bmu^\T  \bbeta\right).
\end{equation}
Note that none of $\pi$, $\bmu^\T \bSigma^{-2}\bbeta$ and $\bmu^\T \bbeta$ can be 0. If $\pi=0$, then (\ref{2}) gives $\bmu^\T \bbeta=0$. Plugging this and $\pi=0$ back into (\ref{1}) gives $\bbeta^\T \bSigma\bbeta\bmu^\T =0$, i.e. $\bmu=0$, which contradicts the assumption $\bmu\neq 0$ and thus $\pi\neq 0$. If $\bmu^\T \bSigma^{-2}\bbeta=0$, then (\ref{2}) gives $\bmu^\T \bbeta=0$. Multiplying $\bSigma^{-2}\bbeta$ from the right to (\ref{1}) and plugging in $\bmu^\T \bSigma^{-2}\bbeta=0$ and $\bmu^\T \bbeta=0$, we have $\bmu^\T \bSigma^{-3}\bbeta=0$ as $\pi\neq 0$. If we repeat this deduction then we have $\bmu^\T \bSigma^{-i}\bbeta=0$, $i=0,1,\cdots$. This can not hold for general $\bmu$ and $\bSigma$ unless $\bbeta=0$, which contradicts $\bbeta=1$, and thus $\bmu^\T \bSigma^{-2}\bbeta\neq 0$. If $\bmu^\T \bbeta=0$, then $\pi=0$ which contradicts $\pi\neq 0$. Therefore (\ref{3}) holds, is well-defined and $\pi\neq 0$. Now plug the $\pi$ back into (\ref{1}) and we have
\begin{equation}\label{4}
\begin{array}{lll}
0  & = & \bbeta^\T \left( \bSigma\bbeta\bmu^\T  \right) - \bbeta^\T  \left(\bmu \bbeta^\T  \bSigma \right) + \left(\bmu^\T \bSigma^{-2} \bbeta\right)^{-1}\left(\bmu^\T  \bbeta\right)\bmu^\T  \bSigma^{-1} \\
& = & \bbeta^\T  U- \bbeta^\T  U^\T  +W, {\rm \  say,}
\end{array}
\end{equation}
and thus $\bbeta^\T  U=\bbeta^\T  U^\T -W$. By multiplying this equation from the right by $\bbeta$ and applying it again, we have
\[
\begin{split}
\bbeta^\T  U \bbeta \ & = \ \bbeta^\T  U^\T \bbeta-W\bbeta \\
& = \ \bbeta^\T  \left(\bbeta^\T  U\right)^\T -W\bbeta \\
& = \ \bbeta^\T  \left( \bbeta^\T  U^\T -W\right)^\T -W\bbeta \\
& = \ \bbeta^\T  U\bbeta-\bbeta^\T  W^\T -W\bbeta,
\end{split}
\]
i.e. $0=W\bbeta+\bbeta^\T  W^\T =\left(\bmu^\T \bSigma^{-2} \bbeta\right)^{-1}\left(\bmu^\T  \bbeta\right)\left[ \bmu^\T  \bSigma^{-1}\bbeta+\bbeta^\T \bSigma^{-1}\bmu \right]$. Since $\bmu^\T \bbeta\neq 0$, we have $ \bmu^\T  \bSigma^{-1}\bbeta+\bbeta^\T \bSigma^{-1}\bmu=0$. Its solution is of the form $\bbeta=\bSigma (V-V^\T )\bmu$ with $V$ a $p\times p$ matrix. Plugging this $\bbeta$ back into (\ref{4}), we have
\[
\begin{split}
0 \  = \ & \bmu^\T (V^\T -V)\bSigma^3(V-V^\T )\bmu\bmu^\T  - \bmu^\T (V^\T -V)\bSigma\bmu\bmu^\T (V^\T -V)\bSigma^2 \\
 & +\left [ \bmu^\T  \bSigma^{-1}(V-V^\T )\bmu  \right]^{-1} \bmu^\T  \bSigma(V-V^\T )\bmu\bmu^\T \bSigma^{-1}.
\end{split}
\]
By the arbitrariness of $\bmu$, the above equation gives
\[
\begin{split}
0 \  = \ & (V^\T -V)\bSigma^3(V-V^\T )\bmu\bmu^\T  - (V^\T -V)\bSigma\bmu\bmu^\T (V^\T -V)\bSigma^2 \\
 & +\left [ \bmu^\T  \bSigma^{-1}(V-V^\T )\bmu  \right]^{-1} \bSigma(V-V^\T )\bmu\bmu^\T \bSigma^{-1}.
\end{split}
\]
Multiplying it by $\bmu^\T  \bSigma^{-2}$ from the left, we have
\[
0 \ = \ \bmu^\T  \bSigma^{-2}(V^\T -V)\bSigma^3(V-V^\T )\bmu\bmu^\T  - \bmu^\T  \bSigma^{-2}(V^\T -V)\bSigma\bmu\bmu^\T (V^\T -V)\bSigma^2+\bmu^\T \bSigma^{-1},
\]
from which, by the arbitrariness of $\bmu$ and $\bSigma$ again, we have
\[
0 \ = \ \bSigma^{-1}(V^\T -V)\bSigma^3(V-V^\T )\bmu\bmu^\T  - \bSigma^{-1}(V^\T -V)\bSigma\bmu\bmu^\T (V^\T -V)\bSigma^2+I.
\]
Thus
\[
I \ = \ \bSigma^{-1}(V^\T -V)\bSigma \left \{\bmu\bmu^\T (V^\T -V)\bSigma^2 -\left[\bmu\bmu^\T (V^\T -V)\bSigma^2 \right]^\T  \right\} \ = \  \bSigma^{-1}(V^\T -V)\bSigma S, {\rm \ say}.
\]
If we take transpose it follows that
\[
\begin{split}
I \ & = \ S^\T  \bSigma (V-V^\T ) \bSigma^{-1}  \\
& = \ -S \bSigma\left[-(V^\T -V)\right ]\bSigma^{-1} \\
& = \ S \bSigma (V^\T -V) \bSigma^{-1}
\end{split}
\]
Combining the above two equations gives $\bSigma^{-1}(V^\T -V)\bSigma S=S \bSigma (V^\T -V) \bSigma^{-1} $ which has the solutions $S=\bSigma^{-2}$ and $S=0$. If $S=\bSigma^{-2}$, then
by the fact $S^\T =-S$ we have $\bSigma^{-2}=-\bSigma^{-2}$ which is not valid. Thus we have $S=0$. From it we have $\bmu\bmu^\T (V^\T -V)\bSigma^2 =\left[\bmu\bmu^\T (V^\T -V)\bSigma^2 \right]^\T $ which has solutions $V^\T -V=\bSigma^{-2}$ and $V^\T -V=0$. If $V^\T -V=\bSigma^{-2}$, then $(V^\T -V)^\T = (\bSigma^{-2})^\T  $, i.e. $V-V^\T =\bSigma^{-2}$, i.e. $-\bSigma^{-2}=\bSigma^{-2}$ a contradiction. Thus $V^\T -V=0$ which gives $\bbeta=\bSigma(V-V^\T )\bmu=0$.

Therefore $\bbeta_{02}=0$ for general $\bmu$ and $\bSigma$. For some specific $\bmu$ and $\bSigma$, there might exist nonzero $\bbeta_{02}$, but always $\bbeta_{02}^\T \bmu=0$. This gives the conclusion that for multivariate normal, the true $d=1$ and only the first $\bbeta_{01}$ contributes to the variation among different classes.

\end{proof}

\subsection{Proof of Theorem \ref{t5}}

\begin{proof}
From $Y\independent\bX|\bB\bX$ and $\bB\bgamma=\mathbf{0}$, we get $Y\independent\bgamma^\T\bX|\bB\bX$ (Proposition 4.3 in \cite{Cook1998}). If we also have $\bgamma^\T\bX\independent\bB\bX$, then by Lemma 4.3 of \cite{Dawid1979} and Proposition 4.6 of \cite{Cook1998}, we obtain $\bgamma^\T\bX\independent Y$ which implies $MV(\bgamma^\T\bX|Y)=0$ according to the property of the MV index. When $\bX\sim N(\bmu,\bSigma)$, then $\bgamma^\T\bX\independent\bB\bX$ if and only if
\be \label{e10}
\cov(\bgamma^\T\bX,\bB\bX)=\bgamma^\T \bSigma\bB^\T=\mathbf{0}.
\ee
If $\bSigma=\mathbf{I}$, then (\ref{e10}) holds by $\bB\bgamma=\mathbf{0}$. Let $\bGamma$ denote the vector space spanned by all $\bgamma$ satisfying (\ref{e10}) and $\bB\bgamma=\mathbf{0}$, then we can easily get $\mathrm{dim}(\bGamma)=p-k$ which implies $d=k$ for $d$ defined in Theorem \ref{t1}. If $\bSigma\ne\mathbf{I}$ and $2k\le p$, $\mathrm{dim}(\bGamma)\ge p-2k$ (note that $\bSigma$ is positive definite), thus $d\le 2k$.
\end{proof}

\subsection{Proof of Theorem \ref{t2}}

\begin{lemma}\label{l1}
Under Condition $(1)$, for any $\bbeta \in \mathbb{R}^p$, we have\\
$(1)$ if $R$ is fixed, $MV_n(\bbeta^{\T}\bX|Y) \to MV(\bbeta^{\T}\bX|Y)$ in probability as $n \to \infty$.\\
$(2)$ if $R$ is diverging with $n$ and satisfies $R = O(n^{\delta})$ with $\delta \le 1/2$,
$MV_n(\bbeta^{\T}\bX|Y) \to MV(\bbeta^{\T}\bX|Y)$ in probability as $n \to \infty$.

\end{lemma}

\begin{proof}
Denote $\bbeta^{\T}\bX$ by $X$ with support $\mathbb{R}_X$ and the transformed samples $\{\bbeta^{\T}\bX_{j}\}_{j=1}^{n}$ by $\{X_{j}\}_{j=1}^n$. By the definitions of $MV(\bbeta^{\T}\bX|Y)$ and $MV_n(\bbeta^{\T}\bX|Y)$, we have
\[
\begin{split}
&~MV_n(\bbeta^{\T}\bX|Y)-MV(\bbeta^{\T}\bX|Y)= MV_n(X|Y) - MV(X|Y)\\
&~= \frac{1}{n}\sum_{j=1}^n \sum_{r=1}^R \hat{p}_r [\hat{F}_{hr}(X_j)-\hat{F}_h (X_j)]^2 - \sum_{r=1}^R p_r \int [F_r(x) - F(x)]^2 d F(x)\\
&~= \sum_{r=1}^R \hat{p}_r \int [\hat{F}_{hr}(x)-\hat{F}_h (x)]^2 d \hat{F}(x) - \sum_{r=1}^R p_r \int [F_r(x) - F(x)]^2 d F(x)\\
&~= \sum_{r=1}^R \hat{p}_r (\int [\hat{F}_{hr}(x)-\hat{F}_h (x)]^2 d \hat{F}(x) - \int [F_r(x) - F(x)]^2 d F(x))\\
&~~~ + \sum_{r=1}^R (\hat{p}_r-p_r) \int [F_r(x) - F(x)]^2 d F(x)\\
&~= \sum_{r=1}^R \hat{p}_r \int ([\hat{F}_{hr}(x)-\hat{F}_h (x)]^2-[F_r(x) - F(x)]^2)d \hat{F}(x) \\
&~~~+ \sum_{r=1}^R \hat{p}_r \int [F_r(x) - F(x)]^2 d [\hat{F}(x)-F(x)]\\
&~~~+ \sum_{r=1}^R (\hat{p}_r-p_r) \int [F_r(x) - F(x)]^2 d F(x)\\
&~=: A_1+A_2+A_3.
\end{split}
\]

For the first term $A_1$,
\[
\begin{split}
|A_1| &~\le 2 \max_{r} \int |[\hat{F}_{hr}(x)-F_r(x)]-[\hat{F}_h(x)-F(x)]|d \hat{F}(x)\\
&~\le 2 \max_{r} \sup_{x \in \mathbb{R}_X}(|\hat{F}_{hr}(x)-F_r(x)|+|\hat{F}_h(x)-F(x)|)\\
&~=:2 (B_1+B_2),
\end{split}
\]
where the second inequality is obtained by $\int d \hat{F}(x) = 1$. We then consider the term $B_1$,
\[
\begin{split}
B_1&~=\max_{r} \sup_{x\in\mathbb{R}_X}|\hat{F}_{hr}(x)-F_r(x)|\\
&~=\max_r O_p(n_r^{-\alpha})\\
&~=O_p(n^{-\alpha(1-\delta)}),
\end{split}
\]
where the second equality is implied by Theorem 2.2 of \cite{Cheng2017} with any $\alpha \le 1/2$ under mild conditions and the last equality is given by Condition $(1)$ and the fact that $|\hat{p}_r-p_r|=O_p(n^{-1/2})$. For the second term, also by Theorem 2.2 of \cite{Cheng2017}, we obtain $B_2=\sup_{x\in\mathbb{R}_X}|\hat{F}_h(x)-F(x)|=O_p(n^{-\alpha})$.

We turn to the term $A_2$,
\[
\begin{split}
|A_2| &~= \sum_{r=1}^R \hat{p}_r \int [F_r(x) - F(x)]^2 d |\hat{F}(x)-F(x)|\\
&~ \le \max_{r} \int [F_r(x) - F(x)]^2 d |\hat{F}(x)-F(x)|\\
&~ \le \int d |\hat{F}(x)-F(x)|\\
&~ \le 2 \sup_{x \in \mathbb{R}_X} |\hat{F}(x)-F(x)|\\
&~ = O_p(n^{-\alpha}),
\end{split}
\]
where the last equality is based on the extended Glivenko-Cantelli lemma (in \cite{Fabian1985}) with $\alpha$ defined above.

For the last term $A_3$, we have $|A_3|=O_p(n^{-\alpha})$ with any $\alpha \le 1/2$ by Lemma A.4 of \cite{Cui2015}. To sum up, $|MV_n(\bbeta^{\T}\bX|Y)-MV(\bbeta^{\T}\bX|Y)|= |A_1+A_2+A_3|\le O_p(n^{-\alpha(1-\delta)}) + O_p(n^{-\alpha})=O_p(n^{-\alpha(1-\delta)})$. Thus, we complete the proof.
\end{proof}

Here, we give a proof for the $d=2$ case of Theorem \ref{t2}.

\begin{proof}
We first prove $\widehat\bbeta_1 \to_{p} \bbeta_{01}$.

From Lemma \ref{l1}, we obtain that for any $\bbeta_1 \in B(\kappa_1)$ with $\kappa_1 \in (0, \kappa_{01}]$, $MV_n(\bbeta_1^{\T}\bX|Y) \to_{p} MV(\bbeta_1^{\T}\bX|Y)$. For any $\epsilon >0$, let $\{\bbeta^{1},\dots,\bbeta^{M}\}$ be an $\epsilon/\sqrt{p}$-net of $B(\kappa_{01})$ with $M=(2\kappa_{01}\sqrt{p}/\epsilon+1)^p$. Since $M$ is fixed, by Lemma \ref{l1} we obtain
\[
\max_{1 \le j \le M}|MV_n(\bbeta^j)-MV(\bbeta^j)|\to_p 0,
\]
as $n \to \infty$. For any $\bbeta \in B(\kappa_{01})$, there exists a $m \in \{1,2,\dots,M\}$ such that $\|\bbeta-\bbeta^{m}\|\le \epsilon/\sqrt{p}$. Then by Conditions $(4)$ and $(5)$, $|MV_n(\bbeta)-MV_n(\bbeta^m)|\le 1/n\sum_{i=1}^n |K_h|_{\infty}\|\bX_i/\sqrt{p}\|\epsilon \to_p |K_h|_{\infty}\E\|\bX/\sqrt{p}\|\epsilon$ and $|MV(\bbeta)-MV(\bbeta^m)|\to_p 0$. Therefore, combining the above three equations, we obtain
\[
\sup_{\bbeta \in B(\kappa_{01})} |MV_n(\bbeta)-MV(\bbeta)|\to_p 0.
\]
Then by Condition $(2.a)$, it is easy to see that for any sufficiently small $\kappa_1>0$,
\[
\mathbb{P}(\sup_{\bbeta_1 \in \partial B(\kappa_1)\cap \Gamma_1} MV_n(\bbeta_1^{\T}\bX|Y)\le MV_n(\bbeta_{01}^{\T}\bX|Y)) \to 1
\]
as $n \to \infty$. Thus, there exists a local maximum local point $\widehat\bbeta_1 \in \partial B(\kappa_1)\cap \Gamma_1$ with probability approaching to $1$ which means
%\[
%\frac{\partial MV_n(\bbeta_1)}{\partial \bbeta_1}|_{\widehat\bbeta_1} = 0
%\] and
$\mathbb{P}(\|\widehat\bbeta_1 - \bbeta_{01}\| < \kappa_1) \to 1$.

Next, Let $p$ be diverging with $n$ and $p^{p/2} n^{-\alpha(1-\delta)}=o(1)$. Note that Lemma \ref{l1} still holds when $p$ is diverging, because the dimension of $\bbeta^{\T}\bX$ stays to be $1$ whether the dimension of $\bX$ diverges or not. Recall that $|MV_n(\bbeta^{\T}\bX|Y)-MV(\bbeta^{\T}\bX|Y)|=O_p(n^{-\alpha(1-\delta)})$ (See Lemma \ref{l1}). Then
\[
\begin{split}
&~\max_{1 \le i \le M} |MV_n(\bbeta_j)-MV(\bbeta_j)| \\
&~ \le \sum_{i=1}^M |MV_n(\bbeta_j)-MV(\bbeta_j)|\\
&~= O_p(Mn^{-\alpha(1-\delta)})\\
&~= o_p(1).
\end{split}
\]
By the same arguments for the fixed $p$ case, we can complete the first part of the proof.

Noting that $\widehat\bbeta_1 \to_{p} \bbeta_{01}$ and Condition $(2.b)$, we can easily obtain $\widehat\bbeta_2 \to_{p} \bbeta_{02}$ by the same arguments given above.

\end{proof}

\subsection{Proof of Theorem \ref{t3}}

%\begin{lemma}\label{l2}
%Under Conditions $(1)$--$(4)$, if $p$ is fixed, we can obtain
%\[
%\sup_{\bbeta_1 \in C(\kappa_{01})} |MV_n(\bbeta_1^{\T}\bX|Y)-MV(\bbeta_1^{\T}\bX|Y)| \to_{p} 0.
%\]
%\end{lemma}

\begin{lemma}\label{l2}
Under Condition $(1)$, for any $\bbeta \in \mathbb{C}^p$, we have\\
$(1)$ if $R$ is fixed, $MV_n(\bbeta^{\T}\bX|Y) \to MV(\bbeta^{\T}\bX|Y)$ in probability as $n \to \infty$.\\
$(2)$ if $R$ is diverging with $n$ and satisfies $R = O(n^{\delta})$ with $\delta \le 1/2$,
$MV_n(\bbeta^{\T}\bX|Y) \to MV(\bbeta^{\T}\bX|Y)$ in probability as $n \to \infty$.
\end{lemma}

\begin{proof}
The proof is similar to that of Lemma \ref{l1}. The only difference comes from the complex $\bbeta$ and the corresponding complex random variable $X=:\bbeta^{\T}\bX$ with support $\mathbb{C}_X$. Denote $X=:a+ib$ where $\bZ:=(a,b)^{\T}$ is a real random vector. By the definition of the cumulative distribution function of a complex variable, that is, the joint distribution function of the real part and the imaginary part of the variable, under several mild conditions we obtain
\[
\sup_{\bZ\in\mathbb{R}^2} |\hat{F}_{h}(\bz)-F(\bz)|= O_{a.s.}(n^{-1/2}(\log{n})^{1/2})\to 0.
\]
This convergence rate comes from Theorem 3 of \cite{LiuYang2008}. Then, by the same arguments of Lemma \ref{l1}, we complete the proof.
\end{proof}

The following is the proof of Theorem \ref{t3}.

\begin{proof}
We first present the proof for $i=1$. For simplicity, we omit $i=1,2$ in the subscripts and superscripts. %Let $\bbeta_1=\bbeta=(b_1,\dots,b_p)^{\T}$
Denote $\btheta=(\bbeta^{\T}, \lambda)^{\T} \in \mathbb{R}^{p+1}$. Let $\widehat\btheta = (\widehat\bbeta^\T, \hat\lambda)^\T$ be the maximizer of $L_{nh}(\btheta)$, then $\widehat\btheta = (\widehat\bbeta^\T, \hat\lambda)^\T$ is a stationary point of $L_{nh}(\btheta)$, that is, $L_{nh}^{'}(\widehat\btheta)=0$. Similarly, let $\btheta_0 = (\bbeta_0^\T, \lambda_0)^\T$ be the maximizer of $L(\btheta)$, then $\btheta_0 = (\bbeta_0^\T, \lambda_0)^\T$ is a stationary point of $L(\btheta)$. Denote $MV(\bbeta) = MV(\bbeta^{\T}\bX|Y)$ and $MV_n(\bbeta) = MV_n(\bbeta^{\T}\bX|Y)$ for the simplicity of notation.

We first prove $n^{1/2}(\widehat\btheta - \btheta_0) \to_{d} N(0,V)$ where the covariance matrix $V$ will be given in the proof below. By the Taylor expansion, we have
\begin{equation} \label{e1}
0=L_{nh}^{'}(\widehat\btheta)=L_{nh}^{'}(\btheta_0) + L_{nh}^{''}(\btheta_0)(\widehat\btheta-\btheta_0)+R(\btheta^{*}),
\end{equation}
where $\btheta^{*}$ satisfies $\|\btheta^{*}-\btheta_0\|\le\|\widehat\btheta-\btheta_0\|$ with $\|\cdot\|$ being the Frobenius norm and $\btheta^{*}=(\bbeta^{*\T},\lambda^{*})^{\T}$. With regular calculation, we obtain
\[
L_{nh}^{'}(\btheta_0)=
\left(\begin{array}{c}
MV_n^{'}(\bbeta_0)+2\lambda_0\bbeta_0\\
\bbeta_0^{\T}\bbeta_0-1
\end{array}\right)
\]
and
\[
L_{nh}^{''}(\btheta_0)=
\left(\begin{array}{cc}
MV_n^{''}(\bbeta_0)+2\lambda_0 I_p & 2\bbeta_0\\
2\bbeta_0^{\T} & 0
\end{array}\right)
\]
where $I_p$ denotes the identity matrix of dimension $p \times p$. The remainder term $R(\btheta^{*})$ contains the third derivative of $L_{nh}(\btheta)$ at $\btheta=\btheta^{*}$. Let $T_n=L_{nh}^{'''}(\btheta^{*})$, where $T_n$ is an array of dimension $(p+1)\times(p+1)\times(p+1)$ and for each $j=1,\dots,(p+1)$ and $T_n(j,;,;)$ is a matrix of dimension $(p+1)\times(p+1)$. Hence, we can write
\[
R(\btheta^{*})=\frac{1}{2}\left(\begin{array}{c}
(\widehat\btheta-\btheta_0)^{\T}T_n(1,;,;)(\widehat\btheta-\btheta_0)\\
(\widehat\btheta-\btheta_0)^{\T}T_n(2,;,;)(\widehat\btheta-\btheta_0)\\
\cdot\\
\cdot\\
\cdot\\
(\widehat\btheta-\btheta_0)^{\T}T_n(p+1,;,;)(\widehat\btheta-\btheta_0)
\end{array}\right)
\]
Then, based on the explicit expressions of the derivatives given above and (\ref{e1}), we obtain
\begin{equation}\label{e2}
\begin{split}
&~-\left(\begin{array}{cc}
MV_n^{''}(\bbeta_0)+2\lambda_0 I_p & 2\bbeta_0\\
2\bbeta_0^{\T} & 0
\end{array}\right)^{-1} \times \sqrt{n}\left(\begin{array}{c}
MV_n^{'}(\bbeta_0)+2\lambda_0\bbeta_0\\
\bbeta_0^{\T}\bbeta_0-1
\end{array}\right)= \\
&~[I_{p+1}+\frac{1}{2} \left(\begin{array}{cc}
MV_n^{''}(\bbeta_0)+2\lambda_0 I_p & 2\bbeta_0\\
2\bbeta_0^{\T} & 0
\end{array}\right)^{-1}\times\left(\begin{array}{c}
(\widehat\btheta-\btheta_0)^{\T}T_n(1,;,;)\\
(\widehat\btheta-\btheta_0)^{\T}T_n(2,;,;)\\
\cdot\\
\cdot\\
\cdot\\
(\widehat\btheta-\btheta_0)^{\T}T_n(p+1,;,;)
\end{array}\right)]\sqrt{n}(\widehat\btheta-\btheta_0).
\end{split}
\end{equation}
%Because the parameter $\lambda$ is not our focus, we avoid it from the last equation to obtain
%\[
%-A_n \times \sqrt{n}(MV_n^{'}(\bbeta_0)+\lambda_0\bbeta_0)=[I_p]
%\]
Next, our proof is divided into two parts:\\
Part 1:
\[
\left(\begin{array}{cc}
MV_n^{''}(\bbeta_0)+2\lambda_0 I_p & 2\bbeta_0\\
2\bbeta_0^{\T} & 0
\end{array}\right)^{-1} \times \sqrt{n}\left(\begin{array}{c}
MV_n^{'}(\bbeta_0)+2\lambda_0\bbeta_0\\
\bbeta_0^{\T}\bbeta_0-1
\end{array}\right)\to N(0,V).
\]
Part 2:
\[
[I_{p+1}+\frac{1}{2} \left(\begin{array}{cc}
MV_n^{''}(\bbeta_0)+2\lambda_0 I_p & 2\bbeta_0\\
2\bbeta_0^{\T} & 0
\end{array}\right)^{-1}\times\left(\begin{array}{c}
(\widehat\btheta-\btheta_0)^{\T}T_n(1,;,;)\\
(\widehat\btheta-\btheta_0)^{\T}T_n(2,;,;)\\
\cdot\\
\cdot\\
\cdot\\
(\widehat\btheta-\btheta_0)^{\T}T_n(p+1,;,;)
\end{array}\right)]\sqrt{n}(\widehat\btheta-\btheta_0)=_{d} \sqrt{n}(\widehat\btheta-\btheta_0),
\]
where $=_d$ indicates convergence in distribution.

{\bf Proof of Part 1.} If $\btheta=(\bbeta^{\T},\lambda)^{\T}$ is a complex vector, the above arguments still stand. By Condition $(4)$ and Lemma \ref{l2}, it follows from the convergence of analytic functions that
\[
L_{nh}^{''}(\btheta_0)^{-1}=
\left(\begin{array}{cc}
MV_n^{''}(\bbeta_0)+2\lambda_0 I_p & 2\bbeta_0\\
2\bbeta_0^{\T} & 0
\end{array}\right)^{-1} \to_p L^{''}(\btheta_0)^{-1}=:A.
\]
%Note that $\bbeta_0^{T}\bbeta_0-1=0$ and $\sqrt{n}\lambda_0 \bbeta_0 =o(1)$, it suffices to show that $MV_n^{'}(\bbeta_0)\to N(0, \bSigma)$ where $\bSigma$ is some covariance matrix.
Denote $\bbeta=(b_1,\dots,b_p)^{\T}$ and $T_{jn}=\sqrt{n}\partial{MV_n(\bbeta)}/\partial{b_j}|_{\bbeta=\bbeta_0}$. Note that $\bbeta_0^{T}\bbeta_0-1=0$, %and $\sqrt{n}\lambda_0 \bbeta_0 =o(1)$,
then it suffices to prove the asymptotic normality of $T_{jn}$ for each $j=1,\dots,p$. It is worth noting the elements of $-A\sqrt{n}L_{nh}^{'}(\btheta_0)$ are linear combinations of $T_{jn}$s. We consider the $j=1$ case.

By the definition of $MV_n(\bbeta)$, we have
\[
T_{1n}=-\sqrt{n} [\frac{1}{n}\sum_{i=1}^{n}\frac{\partial}{\partial{b_1}}\{\sum_{r=1}^{R}\hat{p}_r
[\hat{F}_{hr}(\bbeta^{\T}\bX_i)-\hat{F}_h(\bbeta^{\T}\bX_i)]^2\}]_{\bbeta=\bbeta_0}.
\]
By Lemma \ref{l2},
\[
\sum_{i=1}^{n}\{\sum_{r=1}^{R}\hat{p}_r
[\hat{F}_{hr}(\bbeta^{\T}\bX_i)-\hat{F}_h(\bbeta^{\T}\bX_i)]^2\}=\sum_{i=1}^{n}\{\sum_{r=1}^{R}p_r
[F_{r}(\bbeta^{\T}\bX_i)-F(\bbeta^{\T}\bX_i)]^2\}(1+u_n(\bbeta)+iv_n(\bbeta))
\]
where $\bbeta \in C(\kappa_{0})$, $i^2=-1$, and $u_n(\bbeta)+iv_n(\bbeta)=o_p(1)$ is uniform in $\bbeta\in C(\kappa_{0})$ when $n \to \infty$ with $u_n(\bbeta)$ and $v_n(\bbeta)$ being real functions of $\bbeta$. By Cauchy's residue theorem, we have
\[
T_{1n}=\frac{1}{\sqrt{n}}\frac{1}{2\pi i}\oint_{C_1}\frac{\sum_{i=1}^{n}\sum_{r=1}^{R}\hat{p}_r
[\hat{F}_{hr}(\tilde\bbeta^{\T}\bX_i)-\hat{F}_h(\tilde\bbeta^{\T}\bX_i)]^2}{(b_1-b_{01})^2}db_1,
\]
where $\tilde\bbeta=(b_1,b_{02},\dots,b_{0p})^{\T}$ with $\bbeta_0=(b_{01},\dots,b_{0p})^\T$, and $C_1$ satisfies $\{b_1 \in \mathbb{C}:\|b_1-b_{01}\|=r\}$ with $r<\kappa_0$. Define
\[
S_{1n}=\frac{1}{\sqrt{n}}\frac{1}{2\pi i}\oint_{C_1}\frac{\sum_{i=1}^{n}\sum_{r=1}^{R}p_r
[F_{r}(\tilde\bbeta^{\T}\bX_i)-F(\tilde\bbeta^{\T}\bX_i)]^2}{(b_1-b_{01})^2}db_1.
\]
Then,
\be \label{e3}
T_{1n}-S_{1n}=\frac{1}{\sqrt{n}}\frac{1}{2\pi i}\oint_{C_1}\frac{\sum_{i=1}^{n}\sum_{r=1}^{R}p_r
[F_{r}(\tilde\bbeta^{\T}\bX_i)-F(\tilde\bbeta^{\T}\bX_i)]^2(u_n(\tilde\bbeta)+iv_n(\tilde\bbeta))}{(b_1-b_{01})^2}db_1.
\ee
Let
\[
\frac{1}{\sqrt{n}}{\sum_{i=1}^{n}\sum_{r=1}^{R}p_r
[F_{r}(\tilde\bbeta^{\T}\bX_i)-F(\tilde\bbeta^{\T}\bX_i)]^2}=: R_n(b_1)+iI_n(b_1).
\]
Noting that the left-hand side of (\ref{e3}) is real, we consider the real part of the other hand, that is
\be \label{e4}
\frac{1}{2\pi}\int_{0}^{2\pi}\frac{(R_n \cos\mu+I_n \sin \mu)u_n+(R_n \sin\mu-I_n \cos\mu)v_n}{r}d\mu,
\ee
where the arguments of $R_n,I_n,u_n,v_n$ are $b_{01}+re^{i\mu}$. By the mean value theorem, we obtain
\be \label{e5}
(\ref{e4})=\frac{(R_{0n} \cos\mu_0+I_{0n} \sin \mu_0)u_{0n}+(R_{0n} \sin\mu_0-I_{0n} \cos\mu_0)v_{0n}}{r},
\ee
where the arguments of $R_{0n},I_{0n},u_{0n},v_{0n}$ are all $b_{01}+re^{i\mu_0}$ with $\mu_0 \in [0,2\pi]$. By the definition of $S_{1n}$, we have
\be \label{e6}
S_{1n}=-\frac{1}{\sqrt{n}}\sum_{i=1}^{n}\frac{\partial}{\partial b_1}\sum_{r=1}^{R}p_r
[F_{r}(\bbeta_0^{\T}\bX_i)-F(\bbeta_0^{\T}\bX_i)]^2=-\frac{\partial}{\partial b_1}(R_n(b_{01})+i I_n(b_{01})).
\ee
By central limit theorem, $S_{1n}$ is asymptotically normally distributed. %with mean zero and some variance.
Then, noticing $u_n(\bbeta)+iv_n(\bbeta)=o_p(1)$ for $\bbeta \in C(\kappa_0)$ and letting $r \to 0$, $|T_{1n}-S_{1n}|\le 2/r (|R_{0n}|+|I_{0n}|)(|u_{0n}|+|v_{0n}|)=o_p(1)$ as $n \to \infty$. Using Slutsky's theorem, we complete the proof for the asymptotic normality of $T_{1n}$ as well as Part 1.

{\bf Proof of Part 2.} By Condition $(4)$--$(5)$, we can easily obtain
\[
I_{p+1}+\frac{1}{2} \left(\begin{array}{cc}
MV_n^{''}(\bbeta_0)+2\lambda_0 I_p & 2\bbeta_0\\
2\bbeta_0^{\T} & 0
\end{array}\right)^{-1}\times\left(\begin{array}{c}
(\widehat\btheta-\btheta_0)^{\T}T_n(1,;,;)\\
(\widehat\btheta-\btheta_0)^{\T}T_n(2,;,;)\\
\cdot\\
\cdot\\
\cdot\\
(\widehat\btheta-\btheta_0)^{\T}T_n(p+1,;,;)
\end{array}\right)\to_{p}I_{p+1}.
\]
Hence, by Slutsky's theorem,
\[
[I_{p+1}+\frac{1}{2} \left(\begin{array}{cc}
MV_n^{''}(\bbeta_0)+2\lambda_0 I_p & 2\bbeta_0\\
2\bbeta_0^{\T} & 0
\end{array}\right)^{-1}\times\left(\begin{array}{c}
(\widehat\btheta-\btheta_0)^{\T}T_n(1,;,;)\\
(\widehat\btheta-\btheta_0)^{\T}T_n(2,;,;)\\
\cdot\\
\cdot\\
\cdot\\
(\widehat\btheta-\btheta_0)^{\T}T_n(p+1,;,;)
\end{array}\right)]\sqrt{n}(\widehat\btheta-\btheta_0)=_{d} \sqrt{n}(\widehat\btheta-\btheta_0).
\]

Therefore, combining Part 1 and Part 2, we have $\sqrt{n}(\widehat\btheta-\btheta_0)\to_d N(0,V)$. Because $\widehat\btheta=(\widehat\bbeta^{\T},\hat\lambda)^{\T}$, by the property of multivariate normal distribution, we complete the proof for the case $i=1$ with the covariance matrix $V_1$ being the $p \times p$ sub-matrix at the top right-hand corner of $V$.

For the case $i=2$, the proof is a direct extension of the $i=1$ case. We avoid presenting it here.%Using the conclusion of Theorem \ref{t1} and that of case $i=1$,
\end{proof}

\subsection{Proof of Theorem \ref{t4}}

\begin{proof}
%Assume $\|\nabla MV_n(\bbeta_0)-\nabla MV(\bbeta_0)\|=O_p(\alpha_n)$ where $\alpha_n=\alpha_n(n,p,h)$ and $p^{3/2}\alpha_n =o_p(1)$, $\lambda_{\max} \{\nabla^2 MV_n(\bbeta_0)-\nabla^2 MV(\bbeta_0)\}=o_p(1)$ and .......
It's enough to show that for any given $\varepsilon$ there exists a positive constant $C$ such that, for $n$ large enough, the following inequality holds:
\be \label{e7}
\mathbb{P}\{\sup_{\|\bu\|=C} L_{nh}(\btheta_0+\alpha_n\bu)<L_{nh}(\btheta_0)\} \ge 1-\varepsilon,
\ee
where $\btheta_0$ is the maximizer of $L(\btheta)=MV(\bbeta)+\lambda(\bbeta^\T \bbeta-1)$. This means that with probability tending to $1$ there exists a local maximum $\hat\btheta$ in the ball $\{\btheta_0+\alpha_n\bu:\|\bu\|\le C\}$ so that $\|\hat\btheta-\btheta_0\|=O_p(\alpha_n)$. This implies $\|\hat\bbeta-\bbeta_0\|=O_p(\alpha_n)$. Then,
\[
\begin{split}
&~L_{nh}(\btheta_0+\alpha_n\bu)-L_{nh}(\btheta_0)\\
&=\nabla^\T L_{nh}(\btheta_0)\bu \alpha_n + \frac{1}{2}\bu^\T \nabla^2 L_{nh}(\btheta_0)\bu\alpha_n^2+\frac{1}{6}\nabla^\T\{\bu^\T\nabla^2 L_{nh}(\btheta^*)\bu\}\bu\alpha_n^3\\
&=:B_1+B_2+B_3,
\end{split}
\]
where $\btheta^*$ lies between $\btheta_0+\alpha_n\bu$ and $\btheta_0$.

By Condition $(6)$,
\[
\begin{split}
|B_1|&=|\nabla^\T L_{nh}(\btheta_0)\bu \alpha_n|\le \alpha_n\|\nabla^\T L_{nh}(\btheta_0)\|\|\bu\|\\
&=\alpha_n\|\nabla^\T L_{nh}(\btheta_0)-\nabla^\T L(\btheta_0)\|\|\bu\|\\
&=\alpha_n\|\nabla^\T MV_n(\bbeta_0)-\nabla^\T MV(\bbeta_0)\|\|\bu\|=O_p(\alpha_n^2)\|\bu\|.
\end{split}
\]
For $A_2$, by Condition (7)
\[
\begin{split}
B_2&=\frac{1}{2}\bu^\T \nabla^2 L(\btheta_0)\bu\alpha_n^2+\frac{1}{2}\bu^\T [\nabla^2 L_{nh}(\btheta_0)-\nabla^2 L(\btheta_0)]\bu\alpha_n^2\\
&=-\frac{1}{2}\alpha_n^2|\bu^\T \nabla^2 L(\btheta_0)\bu|+\frac{1}{2}\bu^\T [\nabla^2 MV_n(\bbeta_0)-\nabla^2 MV(\bbeta_0)]\bu\alpha_n^2\\
&=-\frac{1}{2}\alpha_n^2|\bu^\T \nabla^2 L(\btheta_0)\bu|+\frac{1}{2}o_p(1)\alpha_n^2\|\bu\|^2.
\end{split}
\]
By the Cauchy-Schwarz inequality and Conditions $(6)$ and $(8)$,
\[
\begin{split}
|B_3|&=|\frac{1}{6}\sum_{i,j,k=1}^p \frac{\partial L(\btheta^*)}{\partial\theta_i \partial\theta_j \partial\theta_k}u_i u_j u_k \alpha_n^3|\\
&~~~+\frac{1}{6}\nabla^\T\{\bu^\T[\nabla^2 L_{nh}(\btheta^*)-\nabla^2 L(\btheta^*)]\bu\}\bu\alpha_n^3\\
&\le \frac{1}{6} \{\sum_{i,j,k=1}^p M_1^2 \}^{1/2}\|\bu\|^3\alpha_n^3 +\frac{1}{6} \{\sum_{i,j,k=1}^p M_2^2 \}^{1/2}\|\bu\|^3\alpha_n^3\\
&=o_p(\alpha_n^2)\|\bu\|^3.
\end{split}
\]
Allowing $\|\bu\|$ to be large enough, $B_1$ and $B_3$ are dominated by $B_2$ which is less than $0$. This proves (\ref{e7}).

\end{proof}

%%%%%%%%%%%%%%%%%%%%%%%%%%%%%%%%%%%%%%%%%%%%%%%%%%%%%%%%%%%%%%%%%%%%%%%%%%%%%%%%%%%%%%%%%%%%%%%%%%%%%%%%%%%%%%%%%%%%%%%%%%%%%%%%%%%%%%%%%%

\end{document}